\documentclass[10pt]{article}
\usepackage{graphicx}
\usepackage[all]{xy}
\input xy
\usepackage{amsmath}
\usepackage{amsfonts}
\usepackage{amssymb}
\usepackage{url}
\usepackage{multirow}

\usepackage[mathcal]{eucal}
\usepackage{enumerate}
\usepackage{enumitem}
\usepackage{anysize}
\usepackage[english]{babel}
\usepackage{hyperref}
\usepackage{accents}

\usepackage{float}
\usepackage{xcolor}

\newtheorem{example}{Example}
\newtheorem{theorem}{Theorem}
\newtheorem{lemma}{Lemma}

\newenvironment{proof}{\noindent{\bf Proof.}}%
{\hspace*{\fill}$\Box$\par\vspace{4mm}}

\newcommand{\kn}[2] {#1^{\underline{#2}}}
\newcommand{\col} {\mathbf{c}}
\newcommand{\mbar}{\overline{m}}
\newcommand{\pr} {{\rm Pr}}
\newcommand{\pra}[1] {\pr\left\{#1\right\}}
\newcommand{\pd}{\mathbb{P}}
\newcommand{\E} {\mathbb{E}}
\newcommand{\Ea}[1] {\E\left(#1\right)}

\newcommand{\zscore}{$\mathbf{z}$-score}
\newcommand{\dist} {\mathrm{dist}}

\newcommand{\mybreak} {\par\vspace{2mm}\noindent}

\def\cadre{$$\vcenter\bgroup\advance\hsize by -2em\noindent
	\refstepcounter{equation}(\theequation)~\ignorespaces}
\makeatletter
\def\endcadre{\egroup\eqno$$\global\@ignoretrue}

\begin{document}

\title{A novel method for assessing and measuring homophily in networks through second-order statistics}

%\author{Nicola Apollonio\footnote{Istituto per le Applicazioni del Calcolo ``Mauro Picone'', Consiglio Nazionale delle Ricerche, Via dei Taurini 19, 00185 - Rome, Italy . Email: \texttt{nicola.apollonio@cnr.it}.}
%\and 
%Paolo G. Franciosa\footnote{\textbf{Corresponding author} Dipartimento di Scienze Statistiche, Universit\`a di Roma ``La Sapienza'',
%piazzale Aldo Moro 5, 00185 Roma, Italy. Email: \texttt{paolo.franciosa@uniroma1.it}.} 
%\and 
%Daniele Santoni\footnote{Istituto di Analisi dei Sistemi ed Informatica ``Antonio Ruberti'', Consiglio Nazionale delle Ricerche, Via dei Taurini 19, 00185 - Rome, Italy E-mail: \texttt{daniele.santoni@iasi.cnr.it}.}

\author{Nicola Apollonio\textsuperscript{1}
\and 
Paolo G. Franciosa\textsuperscript{2 *}
\and 
Daniele Santoni\textsuperscript{3}
\\
~
\\
\textsuperscript{1} \small{Istituto per le Applicazioni del Calcolo ``Mauro Picone'', Consiglio Nazionale delle Ricerche,}\\
\small{Via dei Taurini 19, 00185 - Rome, Italy . Email: \texttt{nicola.apollonio@cnr.it}.}
\\
\textsuperscript{2 *} \small{\textbf{Corresponding author} Dipartimento di Scienze Statistiche, Universit\`a di Roma ``La Sapienza'',}\\
\small{piazzale Aldo Moro 5, 00185 Roma, Italy. Email: \texttt{paolo.franciosa@uniroma1.it}.} 
\\
\textsuperscript{3} \small{Istituto di Analisi dei Sistemi ed Informatica ``Antonio Ruberti'', Consiglio Nazionale delle Ricerche,}\\
\small{Via dei Taurini 19, 00185 - Rome, Italy E-mail: \texttt{daniele.santoni@iasi.cnr.it}.}
}

\date{}

\maketitle

\begin{abstract}
We present a new method for assessing and measuring homophily in networks whose nodes have categorical attributes, namely when the nodes of networks come partitioned into classes (colors).  We probe this method in two different classes of networks: i) protein-protein interaction (PPI) networks, where nodes correspond to proteins, partitioned according to their functional role, and edges represent functional interactions between proteins ii) Pokec on-line social network, where nodes correspond to users, partitioned according to their age, and edges respresent friendship between users.

Similarly to other classical and well consolidated approaches, our method compares the relative edge density of the subgraphs induced by each class with the corresponding expected relative edge density under a null model.
The novelty of our approach consists in prescribing an endogenous null model, namely, the sample space of the null model is built on the input network itself.
This allows us to give exact explicit expression for the \zscore\ of the relative edge density of each class as well as other related statistics.
The \zscore s directly quantify the statistical significance of the observed homophily via \v{C}eby\v{s}\"{e}v inequality.
The expression of each \zscore\ is entered by the network structure through basic combinatorial invariant such as the number of subgraphs with two spanning edges.
Each \zscore\ is computed in $O(n+m)$ time for a network with $n$ nodes and $m$ edges.
This leads to an overall efficient computational method for assesing homophily. 
We complement the analysis of homophily/heterophily by considering \zscore s of the number of isolated nodes in the subgraphs induced by each class, that are computed in $O(nm)$ time.
Theoretical results are then exploited to show that, as expected, both the analyzed network classes are significantly homophilic with respect to the considered node properties.

\end{abstract}
\vspace{1cm}

%%%%%%%%%%%%%%%%%%%%%%%%%%%%%%%%%%%

\section{Introduction}\label{sec:intro}

The \emph{homophily principle} states that ``similarity breeds connections''\cite{Mcp0}.
This principle---born in sociology---once declined into Network Theory, reads as \emph{nodes in a network are more likely to be linked to  nodes sharing similar attributes}.
The effectiveness of homophily in social networks has been extensively demonstrated across various instances \cite{Auk,Cheadle,karimi,Kossi,Mcp1,Shrum}: social networks exhibit homophily with respect to attributes such as gender, age, ethnicity, occupation, social class and many others.
This simply means that people preferentially interact with people sharing the same cultural and sociological attributes.
Putting it succintely: ``birds of a feather flock together''\cite{Mcp0,eakl}. In contrast, heterophilic networks are those networks whose nodes preferentially interact with nodes having different attributes values. 
Homophily can also be seen as the categorical counterpart of “assortative mixing”---the correlation of attributes across link---and, as such, at least beyond a certain amount of assortativity, it binds the structure of networks~\cite{newman0}, and influences the curvature of the cumulative degree distribution under the preferential attachment evolutionary mechanism\cite{Kim}.
In view of this discussion, homophily qualifies as a genuine network property, namely, a property that when possessed to some extent, impacts non trivially on the structure of the network.
A quantitative understanding of homophily in networks is therefore useful both from a theoretical and a practical point of view.
A step forward in this direction is taken once we realize that homophily in networks certainly fits in the frame of ``community detection"~\cite{lancichinetti2010,gulbache2008}; observe that communities in complex networks identify high order homogeneous structures. 
Arguing as in \cite{yangleskove}, network community detection can be seen as a procedure consisting of two stages: one stage consists of extracting communities by relying on the geometric structure of the networks, while the second stage consists in ``abstracting" communities, namely, in identifying the common features and attributes of community members. Such common features and attributes are usually referred to as \emph{functions} or \emph{node characteristic}\cite{park2007}. 
From this perspective, communities are first detected based on their geometry and then evaluated based on their functions.
The other way round is also meaningful: given a functional description of a network, namely a partition of its nodes into sets of nodes with the same node characteristic, assessing whether or not the class of the partition have a certain amount of geometric structure, i.e.\ assessing whether or not such classes are communities, is tantamount to assessing whether or not the network is homophilic with respect to node characteristic.
Innocent as it may seem, this observation already provides a way of quantifying homophily: if the functional description correlated with the geometry of the network or, equivalently, if the network were homophilic with respect to node characteristics, then the node-induced subgraphs of each class of the partition should be relatively denser than what we expect under some suitable null model---the relative edge density of a subgraph $H$ of a given graph $G$ is the density of $H$ over the density of $G$.
Newman's celebrated \emph{modularity} index \cite{newman} formulates the null hypothesis as the relative expected density of a random graph with the same degree distribution as the input graph. Modularity is thus the relative edge density of monochromatic subgraphs minus the expected relative edge density of monochromatic subgraphs under the null hypothesis that edges are distributed at random among nodes. Since the index lies in interval $[-\frac{1}{2},1]$, its value directly quantifies network homophily: the larger the index the more homophilic  the network is. Modularity thus provides a scale for comparing homophily of different networks. Notice that Newman's index resorts to an exogenous model (the \emph{configuration model}) to test homophily. 
In this paper, by revising the idea in \cite{park2007}, we propose to measure network homophily by testing the observed structure (i.e.\ the relative edge density of functional classes) against the expected structure under an endogeneous random model: the input graph itself will be the sample space for the null hypothesis. 
With this aim in mind, in this paper we propose a new statistical model that builds on the approach in \cite{park2007} (developed for networks with only two functional classes), extend it to an arbitrary number of classes, and strengthen it by exploiting second order statistics based on a uniformly random coloring of the input network with the same color distribution. This machinery yields an explicit exact formula for the \zscore\ of a suitable defined homophily index as well as of the number of isolated nodes of each functional class. The statistical significance of the observed homophily is then obtained through \v{C}eby\v{s}\"{e}v inequality. As one may expect, the structure of the network enters second order statistics through the number of its subgraphs with two spanning edges, namely, the number of its $P_3$'s (if the two edges are adjacent) and the number its $2K_2$'s (if the two edges are not adjacent). This means that our analysis does not require exogenous models (random graphs, for instance) to make comparisons for assessing homophily. Throughout the rest of the paper $P_{3}$ is the graph on three nodes joined by two edges, namely the graph $\xymatrix@C=8pt{
\bullet\ar@{-}[r] & \bullet\ar@{-}[r] & \bullet }$, while $2K_{2}$ is the graph on four nodes  with two edges without common endpoints, namely, the graph $\raisebox{6pt}{\xymatrix@R=-1pt@C=8pt{
\bullet\ar@{-}[r]
& \bullet  \\
\bullet\ar@{-}[r] & \bullet}}.$

In line with the work of \cite{park2007}, we probe our theoretical results on two different network classes: i) Protein-Protein Interaction (PPI) networks, where nodes correspond to proteins, partitioned according to their functional role, and edges represent functional interactions between proteins ii) on-line social network where nodes correspond to users, partitioned according to their age, and edges represent friendship between users. As expected, numerical results provide strong evidence of the homophilic nature of the considered networks with respect to the corresponding node properties: protein function for PPI and age class for social network.

\section{Homophily in networks}\label{sec:homnet} 
As sketched in Section \ref{sec:intro}, we look at homophily as a network parameter (actually as an array of parameters, see Section \ref{sec:measure}) measuring to what extent the attributes (node characteristics, functions) of the nodes of the networks correlate across the edges. To give a precise meaning to such a correlation, we follow the approach in \cite{park2007} which we now discuss in more details. Of course, nothing bad is happening if we think of node characteristics as node colors and, consistently, of the functional description as a partition of the node set into color classes so that (potential) communities are sets of nodes with the same color. Consequently, we deal with a simple undirected graph $G$ with $n$ nodes and $m$ edges whose nodes are partitioned into a number $s$ of color classes.
The simple original model in \cite{park2007} refers to the case of two colors ($s=2$) denoted by $0$ and $1$. Edges of $G$ are then classified as $(0,0)$-edges, $(0,1)$-edges and $(1,1)$-edges according to the color at their endpoints.
Let $c_0$ (resp., $c_1$) be the number of nodes of $G$ having color $0$ (resp., $1$), with $c_{0}+c_{1} = n$; furthermore, let  $m_{i,j}$ be the number of $(i,j)$-edges, $i,\,j\in\{0,1\}$, with $m_{0,0}+m_{0,1}+m_{1,1} = m$: if the functional definition of the communities correlated with the structure of $G$, then we should expect a statistical significant deviation between $m_{0,0}$, say, and what we would expect if characteristic $0$ were randomly distributed  among the nodes of the graph, namely, if any node had an equal chance of possessing it. 
In \cite{park2007}, it is proposed to measure this deviation by the three ratios:
$$\omega_0=\frac{m_{0,0}}{\overline{m}_{0,0}},\quad \eta_{0,1}=\frac{m_{0,1}}{\overline{m}_{0,1}},\quad \omega_1=\frac{m_{1,1}}{\overline{m}_{1,1}}$$
where, for $i,\,j\in\{0,1\}$ and $i\not=j$
$$\overline{m}_{i,i}=m\frac{c_i(c_i-1)}{n(n-1)}\quad \text{and}\quad \overline{m}_{i,j}=m\frac{2c_ic_j}{n(n-1)},$$ 
are the expected number of $(i,i)$-edges and $(i,j)$-edges, respectively, under the hypothesis that properties $0$ and $1$ are randomly distributed  over the node set of $G$ (see Section \ref{sec:main} for proofs).
Just by rewriting $\omega_i$ and $\eta_{i,j}$ as
\begin{equation}\label{eq:reldens}
	\left. \omega_i =\frac{2m_{i,i}}{c_i(c_i-1)}\middle/\frac{2m}{n(n-1)}\right.\quad \text{and}\quad \eta_{i,j}=\left.\frac{m_{i,j}}{c_ic_j}\middle/\frac{2m}{n(n-1)}\right.
\end{equation}
one sees that $\omega_i$ is nothing but the normalized intracommunity density; analogously, $\eta_{i,j}$ is the normalized intercommunity density~\cite{yangleskove}.
In this perspective, homophily (and heterophily) provides a suggestive interpretation of basic structural graph properties (those that can be captured by first order moments of functions of random partitions into two classes with $c_0$ nodes labeled 0 and $c_1$ nodes labeled 1). In this simple model, graph $G$ is declared $i$-homophilic (or homophilic with respect to property $i$), $i\in\{0,1\}$, if $\frac{m_{i,i}}{\overline{m}_{i,i}}>1$; graph $G$ is declared $(i,j)$-heterophilic if $\frac{m_{i,j}}{\overline{m}_{i,j}}>1$ (we ask the reader to bear the pedantic reference to the indices $i,\,j$ in view of the generalization to more than 2 properties).
Without any other clue about the likelihood or the variability of $\omega_i$ and $\eta_{i,j}$, it is clear that  both the assertions have no statistical significance behind their descriptive power. Moreover, it follows from \eqref{eq:reldens} that $\omega_i$ lies in the interval $[0,1/\rho(G)]$, $\rho(G)$ being the edge density of $G$ and such an interval might be really wide for sparse graphs. To overcome this limitation, \cite{park2007} developed a computational model (feasible only for the case of two colors) aimed at evaluating the likelihood of an observed instance $(\omega_0,\eta_{0,1})$ in the form of a phase diagram in the $m_{0,0}m_{0,1}$-plane. Each point of such a diagram is the frequency of all partitions of the node set $G$ into two parts $C_0$ and $C_1$ with $c_0$ and $c_1$ nodes, respectively, such that the subgraph of $G$ induced by $C_0$ has $m_{0,0}$ edges, while the subgraph induced by $C_1$ has $m - (m_{0,0} + m_{0,1})$ edges, $m$ being the size of $G$. The diagram is computed by exhaustive enumeration for small graphs, while for large graphs only the boundary of the diagram is heuristically computed. In either cases, the likelihood of the observed pair is determined by its position and its darkness (in a grayscale) in the phase diagram. Although this approach has been proven successfully for a wide range of real networks (with only two functional classes), including certain PPI networks \cite{park2007}, it still suffers of the following limitations:
\begin{enumerate}[label={\rm (\alph*)}]\itemsep0em
	\item it is computationally expensive. In fact, an exact evaluation of the phase diagram requires time exponential in the number of nodes in the network, and can be applied to large instances only by exploiting heuristic algorithms on a sample. Also, after sampling a subgraph with $\tilde{m}$ nodes, the complexity is  $O(n^{2}\tilde{m})$;
	\item it can be applied only to two functional classes;
	\item it is rather qualitative.
\end{enumerate}
To overcome these limitations, we propose to compute the \zscore\ of $\omega_i$ and $\eta_{i,j}$ under the null model described in the next section. 
Since, as we show, this can be done for any  number $s$ of colors in $O(s(n+m))$ time, and $s$ is usually a small constant, we have that our algorithm is time optimal, hence (a) and (b) are settled. As for (c), if $Z(\omega_i)$, say, is the \zscore\ of $\omega_i$, then by \v{C}eby\v{s}\"{e}v inequality the probability of the event  $(Z(\omega_i)>\lambda)$ is at most $\lambda^{-2}$ under the null model. Hence $Z(\omega_i)^{-2}$ directly measures the statistical significance of $\omega_i$, at the same time making the method completely quantitative. 
Moreover, we propose to evaluate the \zscore\ of the number of isolated nodes in the subgraph induced by each color, that are expected to be negative values in the case of homophilic network. This computation is computationally harder, requiring $O(n m)$ time, but experimental results show to be quite fast on networks with order of $10^{5}$ edges, and is still applicable to sparse networks with about $10^{6}$ nodes.

\subsection{Design of the new model}
Throughout the rest of the paper, we think of a network as an undirected graph $G$ with node-set $V(G)$ and edge-set $E(G)$. An $s$-\emph{coloring} of $G$ is a surjective map $g: V(G)\rightarrow [s]$, where $[s]:=\{1,\ldots,s\}$ is the set of colors. As previously stipulated, we think of $g$ as the functional description of the network, and of the set $g^{-1}(i)$, consisting of the nodes of $G$ having color $i$, as the functional classes of the description. These classes are our (potential) communities. Hence, in the pair $(G,g)$, $G$ encodes the geometrical description of the network and $g$ encodes its functional description. For instance, Protein-Protein Interaction networks (PPI for shortness) are graphs whose nodes are proteins and whose edges model functional interactions between proteins. Since proteins are classified by the biological function they are responsible for, each protein is uniquely associated with one of the 19 functional classes listed in Table \ref{tabfunction} and which we identify by their labels. Therefore, given a PPI network $G$, the correspondence protein$\mapsto$function defines a surjective map $g$ from the set of nodes of $G$ into a set of 19 labels and, after thinking of the labels as colors, such a correspondence will be our 19-coloring $g$. For the Pokec social network graph, we partitioned the node set into five age classes. Therefore a correspondence user$\mapsto$age defines a 5-coloring. Notice that the classification of ages is not frequency based, so that node classes differ substantially in size.

Let $c_i$, $i\in [s]$, be the number of nodes of $G$ of color $i$ under $g$ and call the integer vector $\mathbf{c}=(c_1,\cdots,c_s)$ the \emph{profile} of $g$. Any other coloring $f: V(G)\rightarrow [s]$ with the same profile as $g$ will be referred to as a \emph{$\col$-coloring of $V(G)$} (or simply $\col$-coloring when $V(G)$ is understood). Our next step is to introduce a probability space that allows us to formulate null hypotheses to test against alternative hypotheses about $(G,g)$. To this end, let $\Phi(\col)$ be the set of all $\col$-colorings of $V(G)$. Since the \emph{multinomial coefficient} with \emph{parts} $c_1,c_2\cdots c_s$, denoted by one of the two symbols below
$$\left( n \atop \mathbf{c}\right),\quad \left( n \atop c_1c_2\cdots c_s\right)\ ,$$
counts the $\col$-colorings of $V(G)$ (see the Appendix for a definition of multinomial coefficient), it follows that $|\Phi(\col)|=\left( n \atop \mathbf{c}\right)$. A \emph{random $\col$-coloring} is the random variable $F$ with values in $\Phi(\col)$ and with probability mass function given by
$$\pd_{n,\col}(F)=\pr\{F=f\}=\left( n \atop \mathbf{c}\right)^{-1}\ ,$$
namely, all $\col$-colorings are equally likely (see the Appendix for a more formal definition not needed here). Having the probability space $(\Phi(\col),\pd_{n,\col})$ we test functions of $(G,g)$ versus the same functions under the null hypothesis $(G,F)$, where $F$ is a random $\col$-coloring of $V(G)$. We therefore define  several random variables as functions of the random variable $F$, and such variables enable us to give first and second order moments of those statistics crucial for our purposes. We close this section by describing the former ones, deferring the description of the latter ones to the next section. 
\mybreak
For a node $v\in V(G)$ and a color $i\in [s]$, let $X_v^i$ be the Bernoulli random variable that equals to 1 if and only if node $v$ has color $i$ under the random $\col$-coloring $F$, i.e.\ $X_v^i$ is the indicator of the event $F(v)=i$. Since $X^i_v$ is a Bernoulli random variable, by~(\ref{eq:pipi}) in the Appendix, one has
\begin{equation*}\label{eq:pipi1}
	\Ea{X^i_v}=\pra{X^i_v=1}=\frac{c_i}{n}\ \ .
\end{equation*}
Analogously, for the product of two such variables for $u,\,v\in V$, $u\not=v$, and $i,\,j\in [s]$, after resorting to~(\ref{eq:pipi}) and~(\ref{eq:pipj}) in the Appendix, one has 
\begin{equation}\label{eq:pipj1}
	\Ea{X^i_uX^j_v}=\pra{X^i_uX^j_v=1}=\pra{X^i_u=1, X^j_v=1}=\begin{cases}
		\frac{\kn{c_i}{2}}{\kn{n}{2}} & \text{if $i=j$}\\
		\frac{c_ic_j}{\kn{n}{2}} & \text{if $i\not=j$}
	\end{cases}
\end{equation}
where, after adhering to the notation in \cite{knu}, for a positive integer $a$ and a nonnegative integer $r$, we have denoted by the symbol $\kn{a}{r}$ the \emph{falling $r$-th power} of $a$ (see also the Appendix for more details), namely $\kn{a}{r}=a(a-1)\cdots(a-r+1)$, with $\kn{a}{0}=1$. Thus, the $2$-nd falling power $\kn{a}{2}$ of $a$ equals $a(a-1)$. The above formula  immediately shows that the random variables $X_v^i$ as $v$ runs in $V(G)$ and $i$ runs in $[s]$ are not independent (neither are $X_u^i$ and $X^j_v$). Without pretending to be rigorous, this is only due to the fact that a random $\col$-coloring can be thought of as the outcome of experiments where one draws from a bin ``without replacement". However,  variables in $\{X_v^j \ |\ v\in V, j\in [s]\}$ are \emph{exchangeable}, in the sense that the joint distribution of any subset of them does not depend on the order of drawing (the distribution is symmetric with respect to permuting indices). Hence, as long as we consider statistics based only on linear combinations of $X_v^i$, there is no other dependency other than the one inherited by the sampling procedure. To let the graph come into the structure of the dependency among variables, we have to consider second order statistics.
\mybreak
Let us come to edges now and, for an edge $uv\in E(G)$ and colors $i,\,j\in [s]$, let $Y_{uv}^{i,j}$ be the Bernoulli random variable which is equal to 1 if and only if one of the endpoints of $uv$ has color $i$ and the other one has color $j$. Hence , if $i=j$, then $Y_{uv}^{i,i}=X_u^iX_v^i$ while if $i\not= j$, then $Y_{uv}^{i,j}=X_u^iX_v^j+X_u^jX_v^i$. Therefore by \eqref{eq:pipj1}
\begin{equation}\label{eq:meanedge}
	\Ea{Y_{uv}^{i,j}}=\pra{Y_{uv}^{i,j}=1}=\begin{cases}
		\frac{\kn{c_i}{2}}{\kn{n}{2}} & \text{if $i=j$}\\
		2\frac{c_ic_j}{\kn{n}{2}} & \text{if $i\not=j$}
	\end{cases}\ \ .
\end{equation}
One more random variable is needed to compute the first two moments of the statistics we are interested in. Let $T$ be a nonempty subset of $V(G)$ and let $i\in [s]$ be a color; define $D^i_T$ as the number of elements of $T$ having color $i$; by definition, $D^i_T$ has the following expression:
\[
D^i_T=\sum_{v\in T}X^i_v
\]
Let $A$ and $B$ be disjoint subsets of $V(G)$. To determine the distribution of $D_T^i$ we are interested in the probability of the event that all the elements of $A$ have color $i$ while all those of $B$ have not. Let $\Omega_i(A,B)$ denote this event (for more on events of this type refer to the Appendix). Thus
$$\Omega_i(A,B)=\left(F(a)=i,\,\forall a\in A\right)\wedge\left( F(b)\not=i,\,\forall b\in B\right).$$
Hence
$$\left(D^i_T=h\right)=\bigvee_{\substack{R\subseteq T\\ |R|=h}}\Omega_i(R,T\setminus R)$$ and since the events on the right hand side of the identity above are mutually incompatible, after equation~(\ref{eq:forisolates}) in the Appendix and after setting $t=|T|$, one has 
\begin{equation}\label{eq:sum}
	\pra{D^i_T=h}=\left(t\atop h\right)\frac{\kn{c_i}{h}\kn{(n-c_i)}{t-h}}{\kn{n}{t}}
\end{equation}
and the close resemblance with the binomial distribution with parameters $t$ and $\frac{c_i}{n}$ is clear: powers are replaced by falling powers. This is not an accident: $D^i_T$ follows a hypergeometric distribution $\text{Hyp}(n,c_i,t)$ giving the probability of success by drawing without replacement $t$ balls from an urn containing $n$ balls, $c_i$ of which are successfull. By choosing $T$ equal to the neighborhood of a node $v\in V(G)$, one immediately gets the distribution of the random number of neighbors of node $v$ with color $i$, i.e.\ $D_{N_G(v)}^i\sim\text{Hyp}(n,c_i,\deg_G(v))$.

\subsection{Homophily, heterophily and isolated nodes: first and second order moments}\label{sec:main}
We are now in position to describe statistics capable of assessing whether PPI networks are homophilic. Let $(G,g)$ be a pair consisting of a PPI network $G$ with $n$ nodes and $m$ edges and a $\col$-coloring $g$.
%As we did in Section~\ref{sec:stateofart},
We classify the $m$ edges of $G$ according to the colors of their endpoints. Consequently, we say that edge $uv\in E(G)$ is a $(i,j)$-edge of $(G,g)$ if $\{g(u),g(v)\}=\{i,j\}$, $i,\,j\in [s]$---with a little abuse of notation we also admit $i=j$. Notice that $(i,i)$-edges, the \emph{intra-community} edges, are the edges of $G$ induced by the nodes in  color class $i$ (those responsible for the homophily of $(G,g)$) and, for $i\not=j$, $(i,j)$-edges, the \emph{inter-community} edges, are the edges with one endpoint in color class $i$ and the other one in color class $j$ (those responsible for the heterophily of $(G,g)$). Let $m_{i,i}$ and $m_{i,j}$ be the number of $(i,i)$-edges and $(i,j)$-edges of $(G,g)$, respectively. Therefore, for any two (possibly equal) colors $i,\, j\in [s]$, the random variable
$$M^{i,j}=\sum_{uv\in E(G)}Y_{uv}^{i,j}$$
counts the number of $(i,j)$-edges of $(G,F)$ where $F$ is a random $\col$-coloring. Let $\mbar_{i,j}$ be the expected value of $M^{i,j}$: by \eqref{eq:meanedge} and the linearity of expectation it follows straightforwardly that  
\begin{equation*}
	\mbar_{i,j}=\begin{cases}
		m\frac{\kn{c_i}{2}}{\kn{n}{2}} & \text{if $i=j$}\\
		2m\frac{c_ic_j}{\kn{n}{2}} & \text{if $i\not=j$}
	\end{cases}\ \ ,
\end{equation*}
which generalizes to an arbitrary number of colors the corresponding expressions given above for two colors. Analogously, we define the $i$-homophily of $(G,g)$ and $(i,j)$-heterophily of $(G,g)$, $i\not=j$, as the ratios  
$$\omega_i=\frac{m_{i,i}}{\mbar_{i,i}},\quad \eta_{i,j}=\frac{m_{i,j}}{\mbar_{i,j}},$$
namely, the relative intra- and inter-community density, respectively (recall the identities in \eqref{eq:reldens}). If for all $i,\, j\in [s]$ (possibly $i=j$) we knew the variance $\sigma^2_{i,j}$ of $M^{i,j}$, then we could compute the \zscore\ of the observed  $\omega_i$ e $\eta_{i,j}$ as the ratios
\begin{equation}\label{eq:zscore}
Z(\omega_i)=\frac{m_{i,i}-\mbar_{i,i}}{\sigma_{i,i}}=Z(m_{i,i}),\quad Z(\eta_{i,j})=\frac{m_{i,j}-\mbar_{i,j}}{\sigma_{i,j}}=Z(m_{i,j})\ \ .
\end{equation} 
By \v{C}eby\v{s}\"{e}v inequality, if we assume, for instance, the null hypothesis that the observed value $\omega_i$ is a value assumed by the random variable $\frac{M^{i,i}}{\mbar_{i,i}}$ in the probability space $(\Phi(\col),\pd_{\col,n})$---which is tantamount to assume that $(G,g)$ does not display $i$-homophily---then the confidence level for accepting the null hypothesis would be at most $Z^{-2}(\omega_i)$. Deferring for a while the computation of $\sigma^2_{i,j}$, let us examine another useful statistic for $(G,g)$: the number $l_i$ of isolated nodes in the subgraph induced by color $i$, i.e.\ the number of nodes in color class $i$ having no neighbors in color class $i$. Call any such node $i$-\emph{isolated} and observe that by definition the number of $i$-isolated nodes is
$$l_i=|\left\{v\in V(G) \ | \ g(v)=1\wedge g(w)\not=i,\,\forall w\in N_G(v)\right\}|\ \ .$$
Let $L^i$ be the random variable defined as the number of $i$-isolated nodes of $(G,F)$, where $F$ is a random $\col$-coloring of $V(G)$. Although the random variables $L^i$'s and $M^{i,i}$'s are clearly dependent (as confirmed by results plotted in Fig.~\ref{figcorrelation}) in the next section---at the extreme cases, for instance, $\pra{M^{i,i}=0 \ |\ L^i\geq c_i-1}=1$ and $\pra{M^{i,i}\geq \frac{c_i}{2} \ |\ L^i=0}=1$---the joint knowledge of corresponding statistics $l_i$ and $\omega_i$ is still quite informative. Indeed, consider two graphs $G$ and $\tilde{G}$ on the same node set and let $g$ be a $\col$-coloring of $V(G)$. The $i$-homophily of $(G,g)$ and $(\tilde{G},g)$ could be well the same, but the number of $i$-isolated nodes can be significantly different as in the following example.
\begin{example} For a positive integer $t$ denote by $K_t$ the complete graph on $t$ nodes and by $\overline{K}_t$ its complement, namely the graph with $t$ nodes and no edges. Also denote by $K_{1,t}$ the complete bipartite graph with one node in a color class and $t$ nodes in the other class. Finally, for graphs $G$ and $H$ denote by $G+H$ their disjoint union, namely the graph obtained by picking a copy of $G$ a copy of $H$ disjoint from $G$, and then forming the union of the two copies. Consider the subgraphs $G_i$ and $\tilde{G}_i$ induced by color $i$ in $G$ and $\tilde{G}$, respectively. If, for some positive integer $p$, one has $G_i\cong K_p+\overline{K}_{2p}$ and $\tilde{G}_i\cong K_{p-1}+K_{1,p-1}+\overline{K}_{p}$, then $G_i$ and $\tilde{G}_i$ have the same $i$-homophily but the number of $i$-isolated nodes in $G_i$ is twice the number of $i$-isolated nodes in $\tilde{G}_i$.
\end{example}
Therefore, if we knew that $\omega^i\leq \tilde{\omega}^i$ and $l^i\geq \tilde{l}^i$, then this fact would support the claim that $(\tilde{G},g)$ is more $i$-homophilic than $(G,g)$ because the relative density of property $i$ is less concentrated in $(\tilde{G},g)$ than in $(G,g)$. In conclusion, to assess $i$-homophily of $(G,g)$ the use of the statistics $(Z(\omega^i),Z(l^i))$, where $Z(\omega^i)$ and $Z(l^i)$ are the \zscore s of $\omega^i$ and $l^i$, respectively, could be useful. The next theorem, besides summarizing what we have said about the first order moments of the statistics considered so far, also gives the announced expression for $\sigma^2_{i,j}$ and the expression for the variance of $L^i$. We then exploit these results to compute \zscore s as a tool for analyzing networks in the next section.

\begin{theorem}\label{thm:main}
Let $G$ be a graph with $n$ nodes and $m$ edges and let $(\Phi(\col),\pd_{n,\col})$ be the probability space of the random $\col$-colorings, where $\col=(c_1,\ldots,c_s)$. Assume $c_i>0,\,\forall i\in [s]$. Moreover, let $\pi_3(G)$ denote the number of (not necessarily induced) copies of $P_3$ in $G$. For $i,\, j\in [s]$, consider the random variables $M^{i,j}$ and $L^i$ defined on $(\Phi(\col),\pd_{n,\col})$. Then
\begin{enumerate}[label={\rm\arabic*)}]
\item\label{com:i} for $i\in [s]$ the expected value and the variance of random variable $M^{i,i}=\sum_{uv\in E(G)}Y^{i,i}_{uv}$ where $Y^{i,i}_{uv}=X^i_uX^i_v$ for all $uv\in E(G)$, namely the random number of $(i,i)$-edges of $(G,F)$ under a random coloring $F$, are respectively given by 
\begin{align*}
	\mbar_{i,i}&=m\frac{\kn{c_i}{2}}{\kn{n}{2}},\\
	\sigma^2_{i,i}&=m\frac{\kn{c_i}{2}}{\kn{n}{2}}\left(1-m\frac{\kn{c_i}{2}}{\kn{n}{2}}\right)+2\left\{\left(\frac{\kn{c_i}{3}}{\kn{n}{3}}-\frac{\kn{c_i}{4}}{\kn{n}{4}}\right)\pi_3(G)+\frac{\kn{c_i}{4}}{\kn{n}{4}}\left(m \atop 2\right)\right\};
\end{align*}
\item\label{com:ii}for $i,\, j\in [s]$, $i\not=j$, the expected value and the variance of random variable $M^{i,j}=\sum_{uv\in E(G)}Y^{i,j}_{uv}$ where $Y^{i,j}_{uv}=(X^i_uX^j_v+X^j_uX^i_v)$ for all $uv\in E(G)$, namely the random number of $(i,j)$-edges of $(G,F)$ under a random coloring $F$, are respectively given by
\begin{align*}
	\mbar_{i,j}&=2m\frac{c_ic_j}{\kn{n}{2}},\\
	\sigma^2_{i,j}&=2m\frac{c_ic_j}{\kn{n}{2}}\left(1-2m\frac{c_ic_j}{\kn{n}{2}}\right)+2\left[\left(\frac{c_i\kn{c_j}{2}+\kn{c_i}{2}c_j}{\kn{n}{3}}-4\frac{\kn{c_i}{2}\kn{c_j}{2}}{\kn{n}{4}}\right)\pi_3(G)+4\frac{\kn{c_i}{2}\kn{c_j}{2}}{\kn{n}{4}}\left(m \atop 2\right)\right];
\end{align*}
\item\label{com:iii} for $i\in [s]$ let $L^i$ be the random number of $i$-isolated nodes of $(G,F)$ under a random coloring $F$, namely the random variable $L^i=\sum_{v\in E(G)}W^i_v$, where $W^i_v$ is the Bernoulli variable defined as the indicator of the event $(F(v)=i)\wedge (F(w)\not=i,\,\forall w\in N_G(v))$; then the expected value and the variance of $L^i$ are respectively given by  
\begin{align*}
	\Ea{L^i}=&\frac{c_i}{n}\sum_{v\in V(G)}\frac{\kn{(n-c_i)}{\deg_G(v)}}{\kn{(n-1)}{\deg_G(v)}},\\
	{\rm var}(L^i)&=\Ea{L^i}\left(1-\Ea{L^i}\right)+\frac{\kn{c_i}{2}}{\kn{n}{2}}\sum_{\substack{(u,v)\in V(G)\\ u\not=v, uv\not\in E(G)}}\frac{\kn{(n-c_i)}{b(u,v)}}{\kn{(n-2)}{b(u,v)}},
\end{align*}
where we have set $b(u,v)=|N_G(u)\cup N_G(v)|=\deg_G(u)+\deg_G(v)-|N_G(u)\cap N_G(v)|$. Clearly, $c_i-L^i$ is the random number of nodes of color $i$ spanned by the $(i,i)$-edges.
\end{enumerate} 
\end{theorem}
\noindent
A formal proof of Theorem~\ref{thm:main} is given in the Appendix.

A couple of facts are notable before closing the section.
\mybreak
Statistics presented in points \ref{com:i} and \ref{com:ii} in Theorem~\ref{thm:main}  can be easily computed in $O(n+m)$ time, where $n$ is the number of nodes and $m$ is the number of edges in the input graph, assuming we have a constant number of colors. Hence, computing the $s^{2}$  \zscore s for the number of edges $M^{i,i}$ and $M^{i,j}$ is computationally efficient for any input instance. We observe that the method in~\cite{park2007} requires exponential time for an exact evaluation, or $O(s^2n^3)$ time, where $s$ is the number of functional classes, if optimisation heuristics are exploited.
Computing statistics for the number of isolated nodes $L^{i}$ presented in point \ref{com:iii} in Theorem~\ref{thm:main} is more time consuming. As shown in the Appendix, it requires $O(smn)$ time, that can be improved to $O\left(s \cdot \sum_{v \in V}\deg(v)^{2}\right)$. This is still  efficient for sparse large graphs, with up to millions of nodes and edges.
\mybreak
All of the second order statistics presented in the theorem have an expression that encodes part of the structure of the input graphs, e.g.\ its number of $P_3$'s, $2K_2$'s as well as the cardinalities of the set of common neighbors of nonadjacent pair of nodes. This means that the \emph{coefficient of variation} of
$\omega_i$, defined as $\sigma_{i,i}/\overline{m}_{i,i}$ is completely determined by $G$ and $c_i$ and that different $\col$-colorings (inducing different functional description) have the same scale. In this respect the homophily of the pair $(G,g)$ is an intrinsic measure of the same pair and the coefficient of variation of $\omega_i$ is an invariant of the pair $(G,\col)$. We can thus answer the question ``how homophilic the network is?" without resorting to comparisons with other networks.

\section{Assessing and measuring homophily}\label{sec:measure}
In this section we reap the crops of the last theorem by devising a methodological recipe  to assess and measure homophily in networks. The main tools in this respect are the \zscore s computed in the previous section. Given a pair $(G,g)$ consisting of a network and one of its functional description $g$---a partition of the node-set of the network into $s$ classes of nodes having the same characteristic, e.g.\ age, marital status, biological function, kind of phone subscription, geographical localization etc.---we can define the $s\times s$ random matrix $\mathbf{D}$ whose $i,j$-th entry is the standardized random variable $(M^{i,j}-\mbar_{i,j})/\sigma_{i,j}$ and, analogously, the $s$-dimensional random vector $\mathbf{d}_0$ whose $i$-th entry is the random variable $(L^i-\Ea{L^i})/\sqrt{{\rm var}(L^i)}$---notice that $\mathbf{D}$ is symmetric because $M^{i,j}$ and $M^{j,i}$ are the same variable. From $(G,g)$ we can compute the arrays $\mathbf{Z}$ and $\mathbf{z}_0$ consisting, respectively, of the \zscore s of intra- and inter-community edges (with the former displayed on the main diagonal of the $s\times s$ matrix $\mathbf{Z}$) and of the $s$ \zscore s of the $i$-isolated nodes (nodes of color $i$ none of whose neighbors has color $i$), for $i=1,\ldots,s$. We refer to $\mathbf{Z}$ and $\mathbf{z}_0$ as the \emph{\zscore\ arrays of $(G,g)$}. Hence we may think of $\mathbf{Z}$ and $\mathbf{z}_0$ as the observed values of $\mathbf{D}$ and $\mathbf{d}_0$, respectively---notice that $\mathbf{Z}$ is symmetric as well. For an array $\mathbf{A}$ (matrix or a vector) denote by $1/\mathbf{A}^2$ the array of the same dimensions as $\mathbf{A}$ whose generic entry $b$ is $a^{-2}$, $a$ being the corresponding entry of $\mathbf{A}$. Call the arrays $1/\mathbf{Z}^2$ and $1/\mathbf{z}_0^2$ \emph{$U$-values arrays}. By \v{C}eby\v{s}\"{e}v inequality, the $U$-values arrays give (entry-wise) an upper bound of the probability of observing a value at least as extreme as the one observed for the corresponding random variable. Hence $U$-values are upper bounds of the corresponding $p$-values---so called in the Theory of statistical hypotheses. Although $U$-values arrays:
	\begin{itemize}
		\item  do not capture the statistical dependency structure of the corresponding random arrays---this subject deserves further research;
		\item do not ensure a tight approximation of the corresponding $p$-values: though using only second order moments \v{C}eby\v{s}\"{e}v bounds are undoubtedly the best possible bounds, such bounds can be actually rather loose yielding (possibly) too conservative methods (especially in conjunction with the pervious point),
	\end{itemize}
	$U$-values arrays certainly exhibit the following merits:
	\begin{itemize}
		\item robustness: $U$-values do not require distributional assumptions and therefore have an endogenous nature;
		\item complexity: $U$-values can be efficiently computed (see Section \ref{sec:paolone});
		\item rigour: $U$-values are computed exactly and do not require sampling or estimates and have precise quantitative meaning for homophily.
	\end{itemize}
	Notice that the $U$-values arrays $(1/\mathbf{Z}^2,1/\mathbf{z}^2_0)$ and the \zscore s arrays $(\mathbf{Z},\mathbf{z}_0)$ convey the same statistical information. Hence $(\mathbf{Z},\mathbf{z}_0)$ is already a direct measure of the homophily of $G$ with respect to $g$. We spend the remainder of the section to substantiate this claim.

\paragraph{Descriptive power of \zscore\ arrays and comparisons of networks} The generic entry of $\mathbf{Z}=\{z_{i,j}\}$ measures the distance from the expected value of the corresponding random variables on a scale whose unit is the mean square error. At the same time, such an entry bounds from above the likelihood of this distance through the $U$-values, namely, the map $z_{i,j}\mapsto z^{-2}_{i,j}$. Similar considerations hold for the array $\mathbf{z}_0$.
	It follows that \zscore\ arrays can be conveniently described as heat-maps that provide a visual representation of homophily.  
	These kind of diagrams can be particularly useful when comparing different networks that use the same set of colors because all the arrays involved have the same dimensions and thus the corresponding heat-maps are comparable.
	This can be done for PPI networks, for instance, because they have the same functional description (see Section \ref{sec:ppi} and Section \ref{sec:results}). In this case one can also refine the analysis with the help of vector $\mathbf{z}_0$ to provide a measure of the concentration of homophily in each color class (however we did not pursue this idea numerically).

\paragraph{Multiple Testing} The natural extension of Park and Barabasi's method \cite{park2007} is the following procedure, which we present first in a scalar form to clarify the need for the Bonferroni correction and then in a more algebraic form to confirm the descriptive power of matrix $\mathbf{Z}$. Although in what follows, when dealing with hypothesis testing, it would be more appropriate to use one-sided \v{C}eby\v{s}\"{e}v inequality (a.k.a. Cantelli's inequality)---this amounts to consider $(1+z_{i,j}^2)^{-1}$ in place of $z_{i,j}^{-2}$---for simplicity we stick to the two-sided \v{C}eby\v{s}\"{e}v inequality.

\begin{equation}\label{algo}
\begin{minipage}[c]{370pt}
   \underline{\bf Procedure}. Given the pair $(G,g)$ fix a significance level $\alpha$. Compute the $\mathbf{z}$-scores arrays $(\mathbf{Z},\mathbf{z}_0)$. If $z_{i,i}\geq\frac{1}{\sqrt{\alpha}}$, then declare $G$ $i$-homophilic at level $\alpha$ (recall that $z_{i,i}=Z(\omega_i)$). Analogously, if $z_{i,j}\geq \frac{1}{\sqrt{\alpha}}$, $i\not=j$, then declare $G$ $(i,j)$-heterophilic at level $\alpha$ (recall that $z_{i,j}=Z(\eta_{i,j})$). Array $\mathbf{z}_0$ can be dealt with in the same way and can be used to refine the analysis.
  \end{minipage}
\end{equation}

%\cadre\label{algo} given the pair $(G,g)$ fix a significance level $\alpha$. Compute the \zscore s arrays $(\mathbf{Z},\mathbf{z}_0)$. If $z_{i,i}\geq\frac{1}{\sqrt{\alpha}}$, then declare $G$ $i$-homophilic at level $\alpha$ (recall that $z_{i,i}=Z(\omega_i)$). Analogously, if $z_{i,j}\geq \frac{1}{\sqrt{\alpha}}$, $i\not=j$, then declare $G$ $(i,j)$-heterophilic at level $\alpha$ (recall that $z_{i,j}=Z(\eta_{i,j})$). Array $\mathbf{z}_0$ can be dealt with in the same way and can be used to refine the analysis.  
%\endcadre

\mybreak
While the procedure above correctly assesses homophily (heterophily) of the marginal entries of $\mathbf{D}$, it is not true that the same significance level is valid for the joint distribution of $\mathbf{D}$. For assessing joint homophily (heterophily) we have to look at Procedure~\eqref{algo} as a \emph{multiple testing}  procedure  which therefore requires multiple testing corrections. One of such correction, the most conservative one, is Bonferroni's correction which, in its simplest form, scales level $\alpha$---the level below which the null hypothesis is rejected---by the reciprocal of the number $h$ of testing performed. For instance, suppose we want to assess whether a pair $(G,g)$ is jointly homophillic at level $\alpha$. Then we need to simultaneously test the $s$ diagonal elements of $\mathbf{D}$. In this case, Procedure~\eqref{algo} specializes by declaring that $(G,g)$ is $i$-homophilic when $z_{i,i}>\frac{s}{\sqrt{\alpha}}$. Clearly, as the number of testing increases, the procedure becomes too conservative especially in conjunction with \v{C}eby\v{s}\"{e}v bounds. This limitation is unavoidable without further information about the statistical dependence structure among the marginals of $\mathbf{D}$. Nonetheless, by using a slightly refined form of Bonferroni correction, we can still devise a method to measure homophily in a given network and to compare homophily between different networks that use the same set of colors. For $(i,j)\in [s]\times [s]$, with $i \not= j$, consider the alternative hypothesis $\mathcal{H}_{i,j}^1: D_{i,j}>0$ versus the null hypothesis $\mathcal{H}_{i,j}^0: D_{i,j}\leq 0$ at the significance level $\alpha_{i,j}$. Pair $(i,j)$ is said to \emph{positive at the significance level $\alpha_{i,j}$} whenever Procedure~\eqref{algo} accepts $\mathcal{H}^1_{i,j}$. More generally, for $Q\subseteq \{(i,j)\in [s]\times [s] \ |\ i \not= j\}$, the joint confidence level of the family of tests $\mathcal{H}^0_Q=\{\mathcal{H}^0_{i,j} \ |\ (i,j)\in S\}$---a.k.a \emph{family-wise error rate of the family of tests $\mathcal{H}^0_Q$}---is $\alpha=\min\{1,\sum_Q \alpha_{i,j}\}$ and  set $S$ is called \emph{positive at the joint significance level $\alpha$} whenever $\mathcal{H}_{i,j}^1$ is accepted by Procedure~\eqref{algo} for all $(i,j)\in Q$. The main observation is as follows. If we prescribe the individual significance level $\alpha_{i,j}=z_{i,j}^{-2}$ for $(i,j)\in Q$, then $S$ will be positive at the joint significance level $\min\{1,\sum_Q z_{i,j}^{-2}\}$. In particular, if $Q=\{(i,i) \ |\ i\in [s]\}$, then the set of diagonal positions of $\mathbf{Z}$, namely the positions of the \zscore s of the intra-community densities, is positive at joint confidence level given by the trace of the $U$-value array  $1/\mathbf{Z}^2$. This observation suggests that we can relate the number of positive elements in a set $Q$ at a significance level $\alpha$ with the sum of entries of $1/\mathbf{Z}^2$ indexed by $Q$. Indeed,
%$$\text{let $Q(\alpha)$ be the largest subset of $Q$ such that}\quad \sum_{(i,j)\in Q(\alpha)}z^{-2}_{i,j}\leq \alpha$$
$\text{let $Q(\alpha)$ be the largest subset of $Q$ such that}\quad \sum_{(i,j)\in Q(\alpha)}z^{-2}_{i,j}\leq \alpha$
and let $q(\alpha)$ be the cardinality of $Q(\alpha)$. Notice that $q(\alpha)$ can be 0. Hence $Q$ contains exactly  $q(\alpha)$ positive elements at the joint significance level $\alpha$. Parameter, $q(\alpha)$ depends only on $Q$, $\mathbf{Z}$ and $\alpha$ and therefore can be used to compare different networks that use the same set of colors. On the other hand, by definition, $q(\alpha)$ is related to $\mathbf{Z}$ by the following fact: for a real number $\lambda$, let $J(\lambda)=\{(i,j)\in [s]\in [s] \ |\ i\leq j \wedge z_{i,j}>\lambda\}$. It is clear that for each $\alpha$ there exists a $\lambda$ (not in general unique) such that $Q(\alpha)=J(\lambda)\cap Q$. Therefore, family $\{J(\lambda) \ |\ \lambda\in \mathbb{R}\}$ globally conveys the same information as family $\{q(\alpha) \ |\ \alpha\in [0,1)\}$ and we can get rid of the significance level $\alpha$ when comparing networks that use the same set of colors. Notice however that $\{J(\lambda) \ |\ \lambda\in \mathbb{R}\}$ conveys globally the same information as the heat-map of the \zscore\ matrix $\mathbf{Z}$ with the temperature acting as an inverse transform of the significance level. 

\paragraph{Synthetic measure via Multidimensional \v{C}eby\v{s}\"{e}v-type inequalities}
Multidimensional \v{C}eby\v{s}\"{e}v inequalities~\cite{ferentino} provide a somewhat dual method to the multiple testing procedure above. Recall that if $\mathbf{X}$ is a $d$-dimensional real random vector whose marginals have zero mean and unitary variance, $\|\mathbf X\|$ is the Euclidean norm of $\mathbf{X}$, and $t$ is a positive real number, then the following multidimensional \v{C}eby\v{s}\"{e}v-type inequality holds
$$\pra{\|\mathbf X\|\geq t}\leq \frac{d}{t^2}$$
by a straightforward application of Markov inequality to the random variable $\|\mathbf X\|^2$. The same inequality holds for matrices but replacing the Euclidean norm by the Frobenius norm and adjusting for dimensions. More generally, it holds by vectorializing any subset of entries of a given matrix (after adjusting for dimensions). For instance,  direct application of inequality above yields:
$$\pra{\|{\rm diag}(\mathbf{D})\|\geq \|{\rm diag}(\mathbf{Z})\|}\leq \frac{s}{\|{\rm diag}(\mathbf{Z})\|^2}\ \ ,$$
with ${\rm diag}(\mathbf{A})$ denoting the vector formed by the diagonal entries of the square matrix $\mathbf{A}$. Hence, the sum of the squares of the diagonal entries of $\mathbf{Z}$ gives a global synthetic measure of homophily: the higher such sum is the more globally homophillic the network is. 
Therefore, $$\max\left\{0, 1- \frac{s}{\|{\rm diag}(\mathbf{Z})\|^2}\right\}$$ is a global index of homophily lying in $[0,1]$, like Newman's modularity index~\cite{newman}.

\section{Numerical tests on real networks}\label{MM}
We now probe our theoretical results on two different network classes: i) Protein-Protein Interaction (PPI) networks, where nodes correspond to proteins, partitioned according to their functional role, and edges represent functional interactions between proteins ii) on-line social networks, where nodes correspond to users, partitioned according to their age, and edges represent friendship between users. 
As shown in the previous section, the major character of our methodology is the \zscore\ matrix $\mathbf{Z}$.
Let us discuss data and the running time of the method in some details before going to the numerical tests.

\subsection{Protein-protein interaction networks}\label{sec:ppi}
We consider ten PPI networks retrieved from STRING database (https://string-db.org/) \cite{szklarczyk2017,szklarczyk2019}, setting a high confidence score cut-off (0.70). 
The selected networks, listed in Table~\ref{taborganisms}, are mainly related to Bacteria (8 out of 10, belonging to diffent Phyla or classes), we also included in the study {\it Saccharomices cerevisiae} (Fungi - Ascomycota) and {\it Pyrococcus abyssi} (Euryarcheota - Thermococci) for comparison. 
The 8 bacterial organisms were chosen as representatives of Bacteria Kingdom, including different Phyla (Alpha, Gamma, Epsilon proteobacteria, Actinobacteria, Firmicutes/Bacilli, Spirochaetes). Organisms were also chosen on the basis of their network sizes (number of nodes and edges), in order to build an etherogeneous dataset.  
Species, Kingdom, Phylum/Class as well as number of nodes, number of edges, and density of the relative network are reported in Tab.~\ref{taborganisms} for each organism.
\begin{center}

\begin{table}[H]
\scriptsize
\begin{tabular}{|c|c|c|r|r|r|}
\hline
\multicolumn{3}{|c|}{\textbf{Organism}}&\multicolumn{3}{|c|}{\textbf{PPI network}}\\
\hline
\textbf{Species}&\textbf{Kingdom}&\textbf{Phylum/Class}&\multicolumn{1}{|c|}{\textbf{nodes}}&\multicolumn{1}{|c|}{\textbf{edges}}&\multicolumn{1}{|c|}{\textbf{{density}}}\\
\hline
{\it Brucella melitensis} \textbf{(Bm)} &Bacteria&Alphaproteobacteria&2,675&15,450&{0.43\%}\\
\hline
{\it Escherichia coli} \textbf{(Ec)} &Bacteria&Gammaproteobacteria&4,020&29,748&{0.37\%}\\
\hline
{\it Haemophilus influenzae} \textbf{(Hi)} &Bacteria&Gammaproteobacteria&1,609&9,202&{0.71\%}\\
\hline
{\it Helicobacter pylory J99} \textbf{(Hp)} &Bacteria&Epsilonproteobacteria&1,264&7,678&{0.96\%}\\
\hline
{\it Mycobacterium tuberculosis H37Rv} \textbf{(Mt)} &Bacteria&Actinobacteria&3,779&24,889&{0.35\%}\\
\hline
{\it Streptococcus pneumoniae TIGR4} \textbf{(Sp)} &Bacteria&Firmicutesi/Bacilli&1,811&8,813&{0.54\%}\\
\hline
{\it Treponema pallidum} \textbf{(Tp)} &Bacteria&Spirochaetes&894&8,157&{2.04\%}\\
\hline
{\it Vibrio cholerae} \textbf{(Vc)} &Bacteria&Gammaproteobacteria&3,153&20,844&{0.42\%}\\
\hline
{\it Pyrococcus abyssi} \textbf{(Pa)} &Euryarchaeota&Thermococci&1,564&9,090&{0.74\%}\\
\hline
{\it Saccharomyces cerevisiae} \textbf{(Sc)} &Fungi&Ascomycota/Saccharomycetes&6,157&119,051&{0.63\%}\\
\hline

\hline
\end{tabular}
\caption{Complete list of considered organisms, together with their network size (nodes and edges). Density is expressed as the ratio between the actual number of edges and the number of edges in the complete graph with the same number of nodes.}
\label{taborganisms}
\end{table}
\end{center}
\vspace{0.2truecm}
Functional classes of proteins of the considered ten organisms were obtained from NCBI database\\
(ftp://ftp.ncbi.nih.gov/pub/COG/COG/).
Proteins were partitioned into 25 different functional classes, but only 19  were taken into account in this work, since:
\begin{itemize}
\item 5 classes (A - RNA processing and modification, B - Chromatin structure and dynamics, Y - Nuclear structure, Z - Cytoskeleton, W - Extracellular structures) had no representatives (or only a few) for most of bacterial organisms;
\item  classes R -  general function prediction, and S - Function unknown,  were merged  into the X class.
\end{itemize}
The 19 considered classes are reported in Tab~\ref{tabfunction}. The number of proteins for each functional class in each organism is reported in the Appendix.

\begin{center}
\begin{table}[H]
\begin{tabular}{|c|l|}
\hline
\multicolumn{2}{|c|}{\textbf{INFORMATION STORAGE AND PROCESSING}}\\
\hline
%\textbf{Class}&\textbf{Description}\\
%\hline
J&Translation, ribosomal structure and biogenesis\\
\hline
%A&RNA processing and modification\\
%\hline
K&Transcription\\
\hline
L&Replication, recombination and repair\\
\hline
%B&Chromatin structure and dynamics\\
%\hline
\multicolumn{2}{|c|}{\textbf{CELLULAR PROCESSES AND SIGNALING}}\\
\hline
D&Cell cycle control, cell division, chromosome partitioning\\
\hline
%Y&Nuclear structure\\
%\hline
V&Defense mechanisms\\
\hline
T&Signal transduction mechanisms\\
\hline
M&Cell wall/membrane/envelope biogenesis\\
\hline
N&Cell motility\\
\hline
%Z&Cytoskeleton\\
%\hline
%W&Extracellular structures\\
%\hline
U&Intracellular trafficking, secretion, and vesicular transport\\
\hline
O&Posttranslational modification, protein turnover, chaperones\\
\hline
\multicolumn{2}{|c|}{\textbf{METABOLISM}}\\
\hline
C&Energy production and conversion\\
\hline
G&Carbohydrate transport and metabolism\\
\hline
E&Amino acid transport and metabolism\\
\hline
F&Nucleotide transport and metabolism\\
\hline
H&Coenzyme transport and metabolism\\
\hline
I&Lipid transport and metabolism\\
\hline
P&Inorganic ion transport and metabolism\\
\hline
Q&Secondary metabolites biosynthesis, transport and catabolism\\
\hline
\multicolumn{2}{|c|}{\textbf{POORLY CHARACTERIZED}}\\
\hline
X&Function unknown or general function prediction only\\
\hline

\end{tabular}
\caption{Protein functional classes, partitioned into higher categories.}
\label{tabfunction}
\end{table}
\end{center}
\vspace{0.2truecm}

Each organism's network is an undirected graph, in which each node represents a protein associated to a color denoting one of the functional classes listed in Table~\ref{tabfunction}, and each edge represents the interaction between two proteins, weighted according to the likelihood of the given interaction. A PPI graph is thus represented by two text files, the first lists node labels and the associated colors, the second lists edges as pairs of nodes and the associated weight in range $[0, 999]$. Edges have been cut-off at a 700 minimum weight, usually considered as a high confidence threshold. Isolated nodes in the resulting graph have been deleted. 
Some networks present a very limited number of nodes (some units) labeled by similar values (e.g.\ \texttt{jhp0681\_1} and \texttt{jhp0681\_2} in the node file for organism {\it Helicobacter pylori}) representing different isoforms of the same protein, but these nodes were simply denoted by a unique label (e.g.\ \texttt{jhp0681}) in the edge listing file. We merged such nodes in a single node; in the few cases in which they were associated to different functional classes, we merged them associating the functional class X to that node.
%Our experiments have been conducted on an Intel Core i5 machine with 4 cores, 2.3 GHz clock, 16 GB RAM, 256 KB L2 cache and 6 MB L3 cache, equipped with  MAC OS 10.14.6.

%All methods have been written in Python 3 and executed by a Python 3.7.2 interpreter, exploiting package {\it NetworkX}, version 2.5.

\subsection{Pokec social network}
Pokec is the most popular Slovak on-line social network.
Datasets, obtained during May 25-27 2012, are anonymized and contain relationships and user profile data of the whole network \cite{takac2012}. 
Friendships in the Pokec network are originally oriented. We decided to consider only symmetric pairs, so that we derived an undirected graph where nodes $x,y$ are adjacent if and only if both $x$ is a friend of $y$ and $y$ is a friend of $x$, so that it can be assessed that the two considered members had an actual interaction; also in this case, isolated nodes have been discarded. The network obtained contains more than one million nodes and 8 millions edges. Nodes are partitioned in classes according to the age declared by members, where about 34\% of them either did not declare age, or declared a patently untrue value---in some cases even less than 10 or over 100. So, we decided to put   into a ``fake'' age class denoted by $X$ all members whose age is not a numeric value in $[12,60)$.  The size of each subgraph induced by the 5 age classes, possibly containing isolated nodes, is shown in Table \ref{tab:pokec}, together with the size of the entire network.

%Pokec is the most popular Slovak on-line social network.
%Datasets, crawled during MAY 25-27 2012, are anonymized and contains relationships and user profile data of the whole network \cite{takac2012}. 
%Friendships in the Pokec network are originally oriented, only interaction between nodes with bidirectional edges were considered, so that it can be assessed that the two considered nodes had an actual interaction.

\begin{table}[H]

\begin{center}
%\small
\begin{tabular}{|c|c|r|r|}
\hline
\textbf{Class}&\textbf{Age}&\multicolumn{1}{|c|}{\textbf{Nodes}}&\multicolumn{1}{|c|}{\textbf{Edges}}\\
%\textbf{Class}&\textbf{Age}&\textbf{Nodes}&\textbf{Edges}\\
\hline
C&[12,18)&152,659&348,617\\
\hline
D&[18, 25)&332,826&2,038,089\\
\hline
E&[25, 40)&270,299&521,228\\
\hline
F&[40, 60)&46,295&23,156\\
\hline
X&Otherwise&410,270&949,026\\
\hline
\multicolumn{2}{|c|}{\textbf{whole network}} & \textbf{1,212,349} & \textbf{8,320,600}\\
\hline

\end{tabular}
\caption{Classes of Pokec social network. For each class the number of nodes is reported, with the number of edges joining nodes in the same class.}
\label{tab:pokec}

\end{center}
\end{table}
\vspace{0.2truecm}

\subsection{Implementation details}\label{sec:paolone}
We developed a Python 3 prototype implementing our model, source code is available at\\
\texttt{http://www.statistica.uniroma1.it/users/pfrancio/homophily/}\\
Experiments have been performed on an Intel Core i5 PC with 4 cores, 2.3 GHz clock, 16 GB RAM, 256 KB L2 cache and 6 MB L3 cache, equipped with  MAC OS 10.14.6. For the huge Pokec network, a 250 GB RAM machine running 18.04.5 LTS has been used.

Computing times, using a single core, are reported in Table \ref{tab:times}, excluding time elapsed in file I/O. 
As it clearly appears from the table, the ratio between the number of edges in the graph and the time needed to compute edge \zscore s is close to be constant (varying from 340k to 460k edges per second), confirming the asymptotic complexity $O(n+m)$---assuming the number of colors is constant.

An efficient computation of singleton \zscore s requires some more care. Expression for ${\rm var}(L^i)$ in point \ref{com:iii}) in Theorem~\ref{thm:main} requires $O(n^{3})$ time to be computed. Actually, it can be manipulated (details are discussed in the Appendix), so that the complexity of computing ${\rm var}(L^i)$ for each color $i$ is lowered to $O(nm)$. More precisely, its complexity is strictly related to the number of pairs of  nodes at distance 2, which in turns is bounded by $\pi_{3}$, i.e.\ the number of $P_{3}$'s in the graph. It is immediate to see that
$$\pi_{3} = \frac{1}{2}\sum_{v \in G}(\deg_{G}(v))^{\underline{2}}  \leq \frac{1}{2}\sum_{v \in G}(\deg_{G}(v))^{2}$$
The sum of squared degrees for all experimented networks is reported in Table  \ref{tab:times}, where it is confirmed to be proportional to  computing times for singleton \zscore s (with a ratio varying from 28k to 61k $P_{3}$'s per second).

\begin{center}

\begin{table}[H]
\small
\begin{center}

\begin{tabular}{|c|r|r|r|r|r|}

\hline

\multicolumn{1}{|c|}{\textbf{Network}}&\multicolumn{3}{|c|}{\textbf{size}}&\multicolumn{2}{|c|}{\textbf{computing time (seconds)}}\\

\hline

&\textbf{nodes}&\textbf{edges}&\textbf{squared degrees sum}&\textbf{edge \zscore}&\textbf{singleton \zscore}\\

\hline
\hline

\textbf{Bm} & 2,675&15,450&942,470&0.042&15.338\\

\hline

\textbf{Ec} &4,020&29,748&1,947,532&0.077&63.174\\

\hline

\textbf{Hi} &1,609&9,202&607,128&0.023&10.477\\

\hline

\textbf{Hp} &1,264&7,678&535,246&0.020&9.973\\

\hline

\textbf{Mt} &3,779&24,889&1,574,806&0.068&43.241\\

\hline

\textbf{Sp} &1,811&8,813&555,570&0.023&9.010\\

\hline

\textbf{Tp} &894&8,157&818,544&0.021&14.284\\

\hline

\textbf{Vc} &3,153&20,844&1,505,448&0.054&39.030\\

\hline

\textbf{Pa} &1,564&9,090&713,514&0.022&12.510\\

\hline

\textbf{Sc} &6,157&119,051&30,075,870&0.257&1,062.981\\

\hline

\textbf{Pokec}&1,212,349&8,320,600&752,382,968&24.270&24,086.467\\

\hline

\end{tabular}

\caption{Computing times for edge \zscore s and singleton \zscore s, on organisms and Pokec networks. For each network we  report the number of nodes, the number of edges and the sum of squared degrees. The complexity of singleton \zscore s computation strongly depends on the sum of squared degrees.}

\label{tab:times}

\end{center}

\end{table}

\end{center}

\vspace{0.2truecm}

\subsection{Numerical results}\label{sec:results}
In order to have a pictorial quantitative perception of homophily and heterophily in the considered networks, we present matrix $\mathbf{Z}$ of the \zscore s of the intra- and inter-community edges (see Section \ref{sec:measure}) in the form of heat-maps. Color scale is logarithmic on \zscore s, traslated in order to avoid negative values. Each entry of $\mathbf{Z}$ corresponds to a square in the diagram. Green squares corresponding to entry $i,j$ represent positive \zscore s, while pink squares represent negative \zscore s. 
Results related to PPI networks are shown in Figure~\ref{fig:heat-mapsPPI}. Homophily of PPI's with respect to their functional description is clearly readable from all the heat-maps by the green squares in all diagonals---showing the relative intra-community density---except for the poorly characterized X function class. A majority of off-diagonal \zscore s are negative (more than 79.6\%), while diagonal \zscore s tend to show very high values.
As a global result, neglecting all $i,j$ pairs where either
$i = \mathrm{X}$ or $j = \mathrm{X}$, we recap that:
\begin{itemize}
\item the average value of the diagonal entries $\mathbf{Z}$ is 36.26,
with standard deviation 49,85, ranging from a  -0.3183 minimum to a 326.6 maximum;
\item more than 91\% of diagonal entries $\mathbf{Z}$ are greater than 5;
\item the average value of off-diagonal entries $\mathbf{Z}$ is -0.836,
with standard deviation 3.707, ranging from a  -5.983 minimum to a 55.67 maximum;
\item more than 65\% of off-diagonal entries $\mathbf{Z}$ are less than -1.
\end{itemize}

\begin{figure}[H]
\centerline{
\includegraphics[scale=0.4]{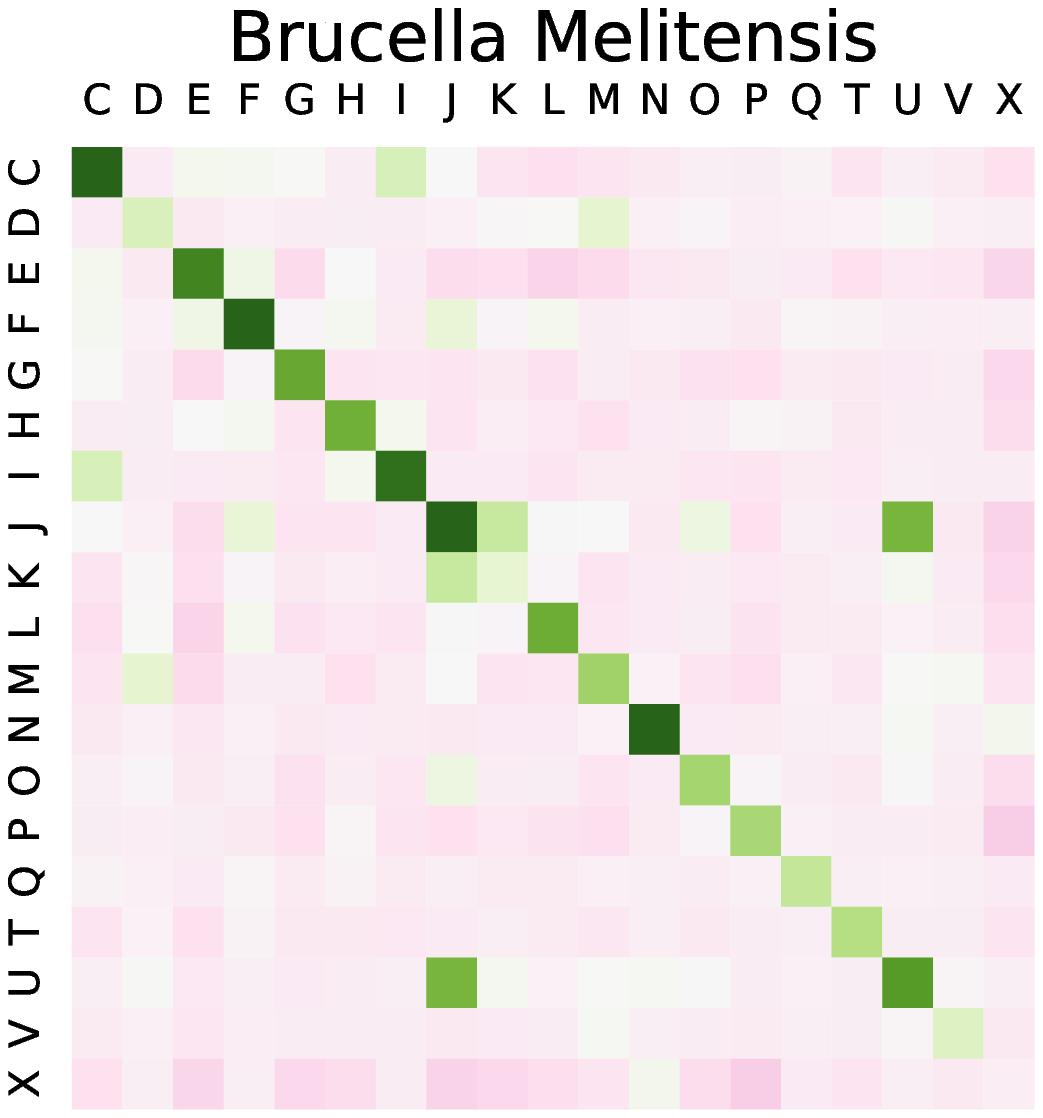}
\includegraphics[scale=0.4]{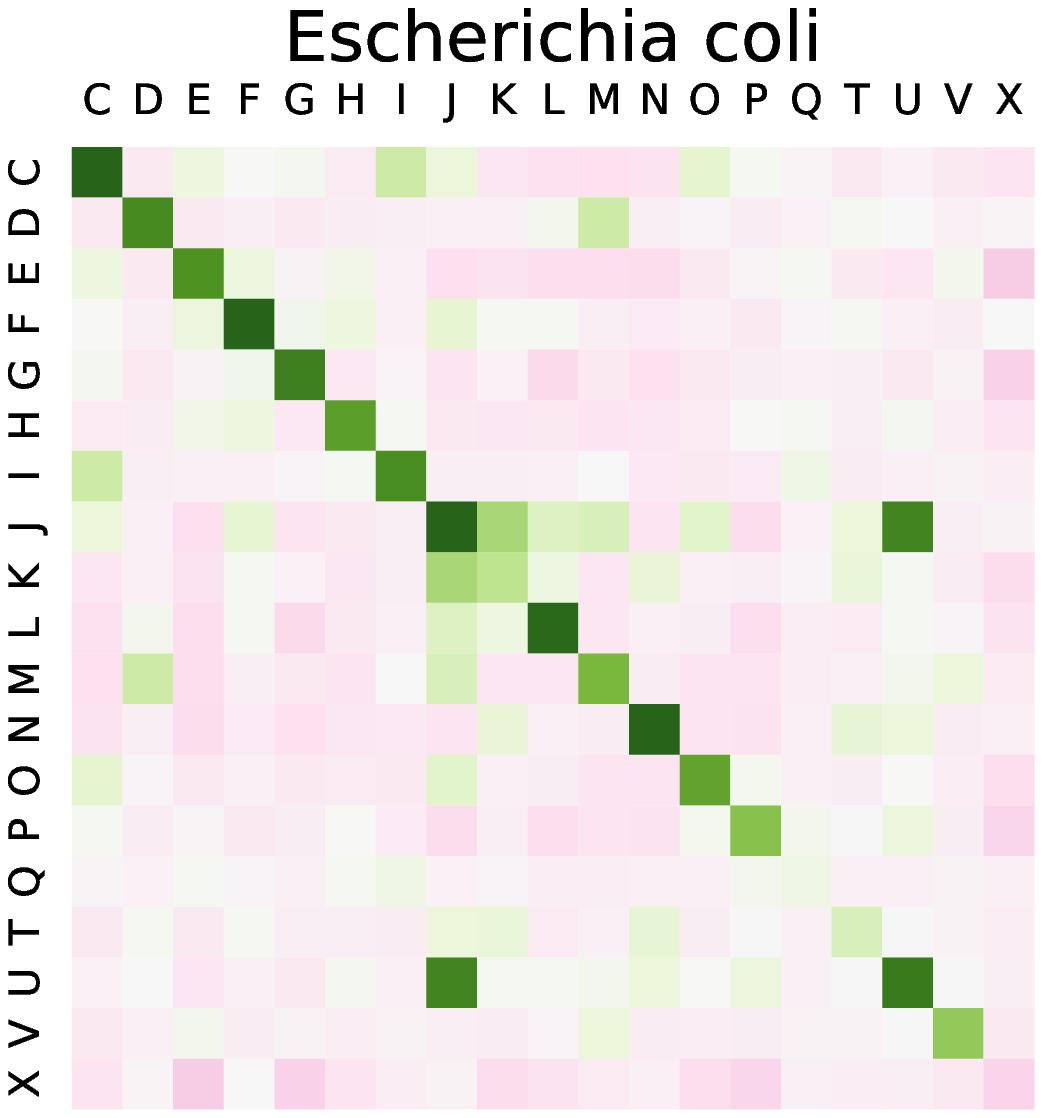}
\includegraphics[scale=0.4]{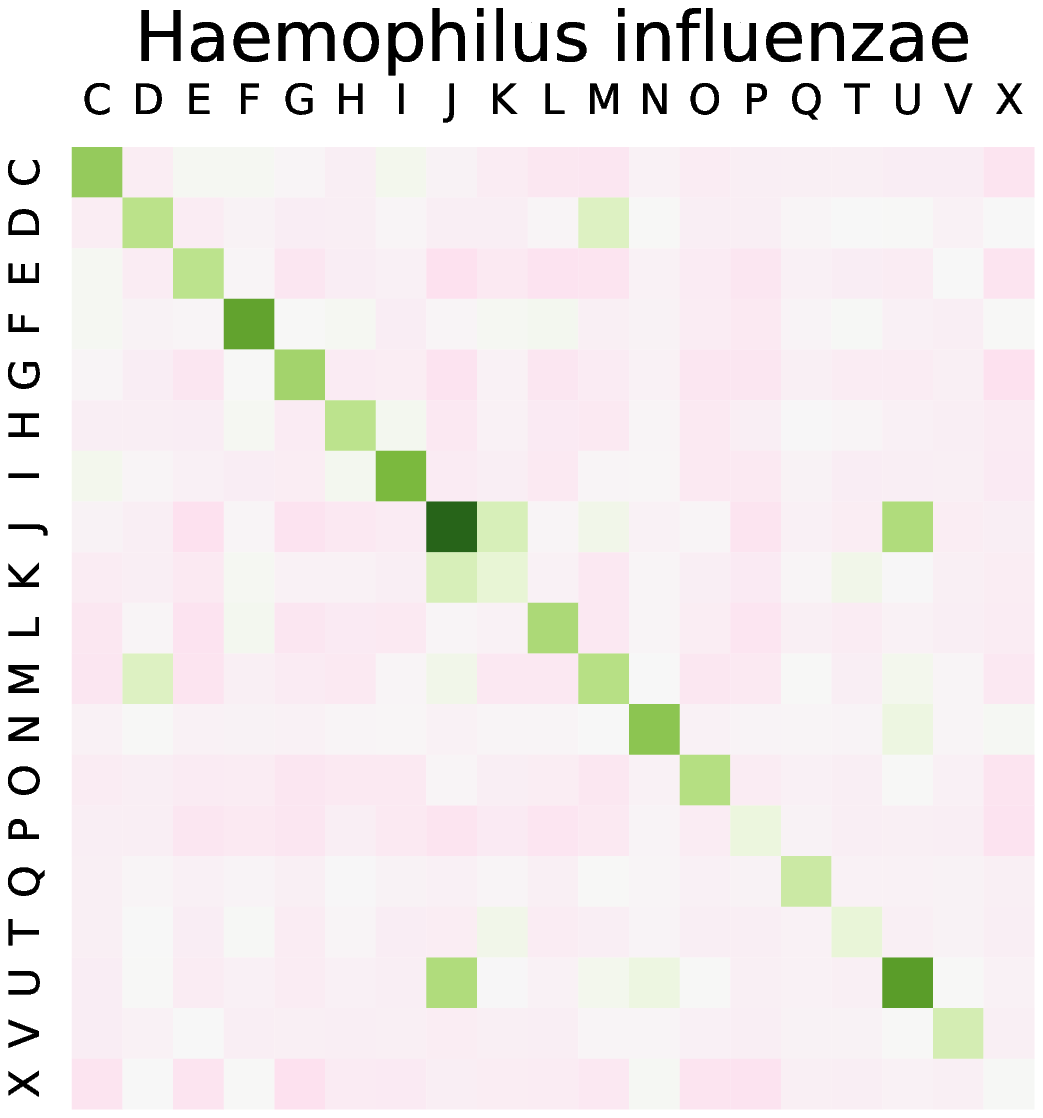}
\includegraphics[scale=0.3]{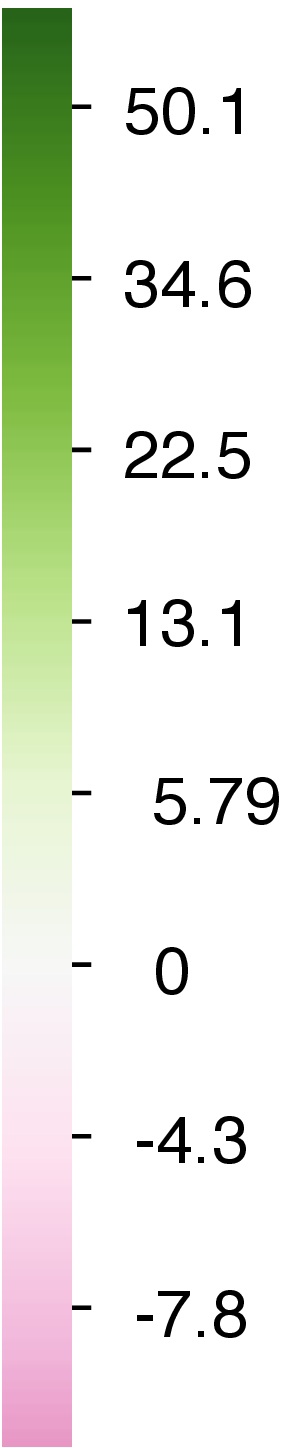}
}

\centerline{
\includegraphics[scale=0.4]{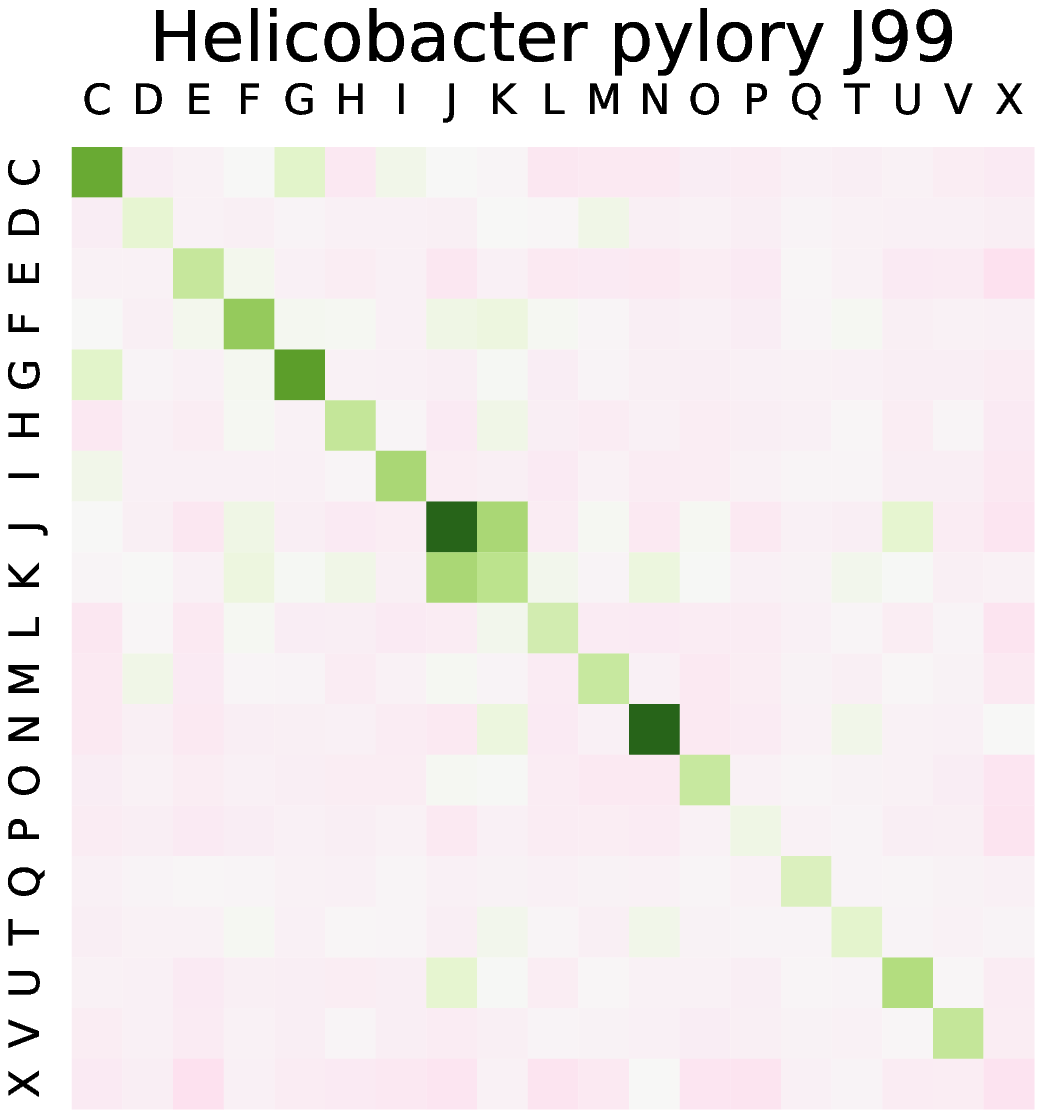}
\includegraphics[scale=0.4]{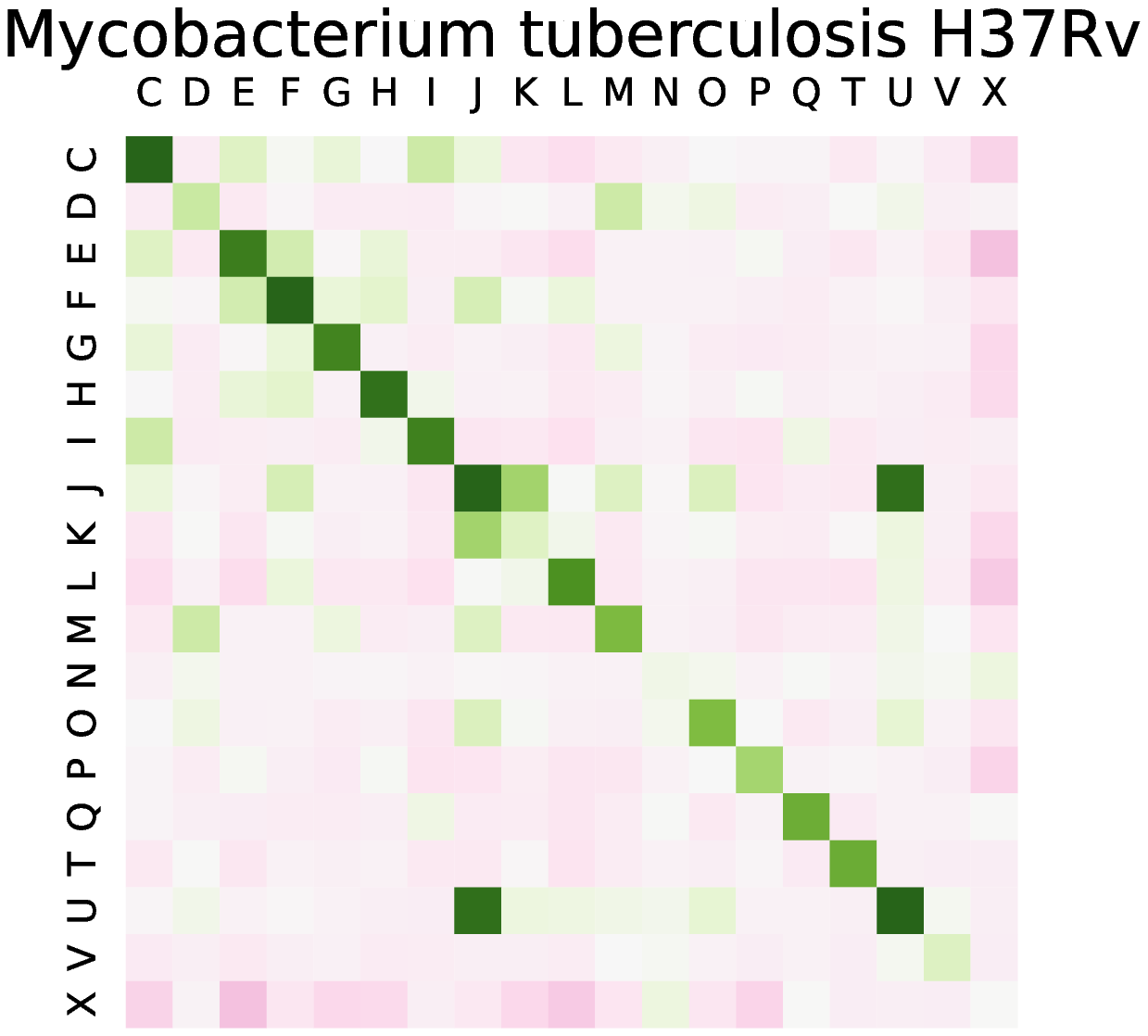}
\includegraphics[scale=0.4]{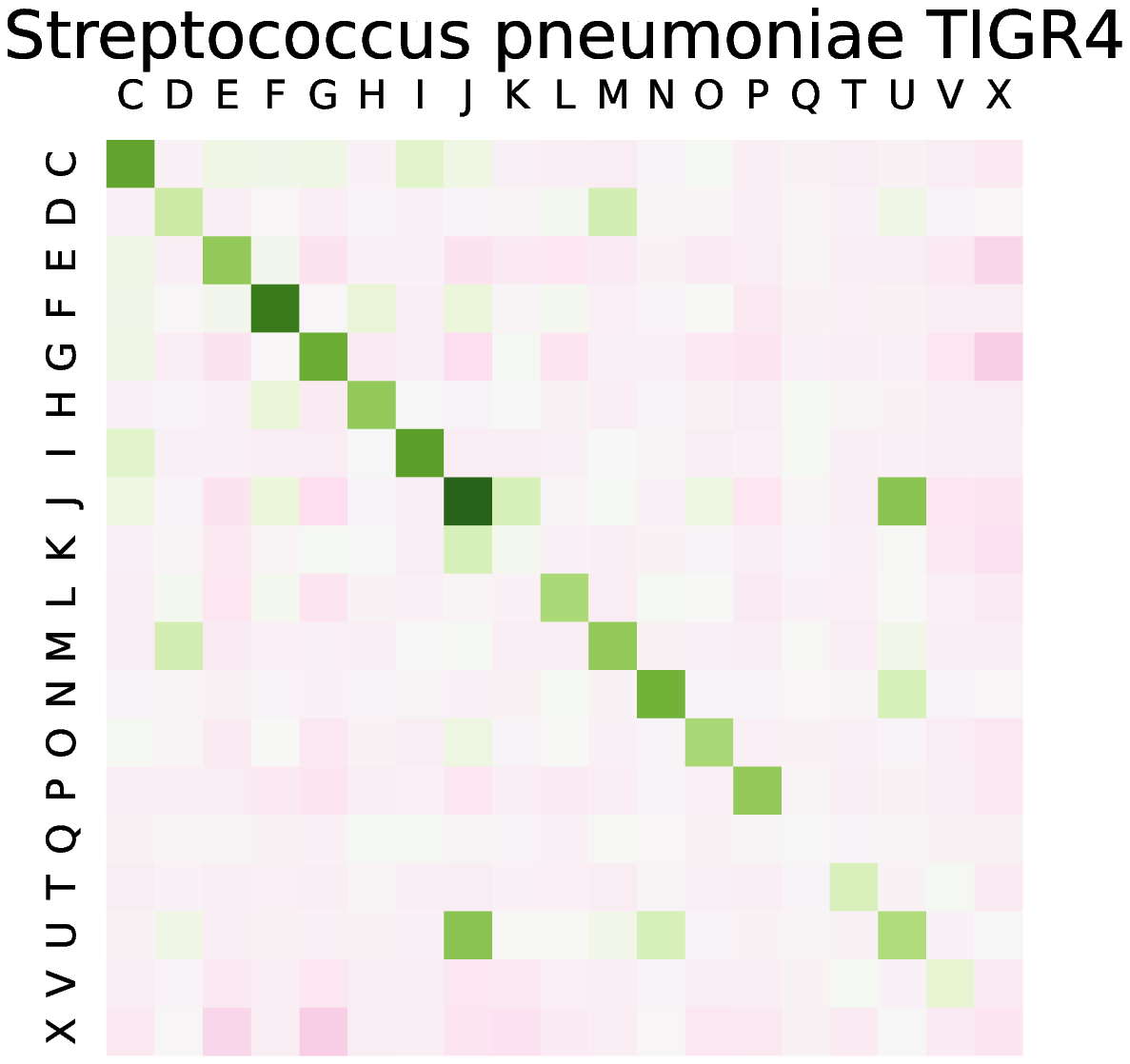}
\includegraphics[scale=0.3]{colorbar.eps}
}

\centerline{
\includegraphics[scale=0.4]{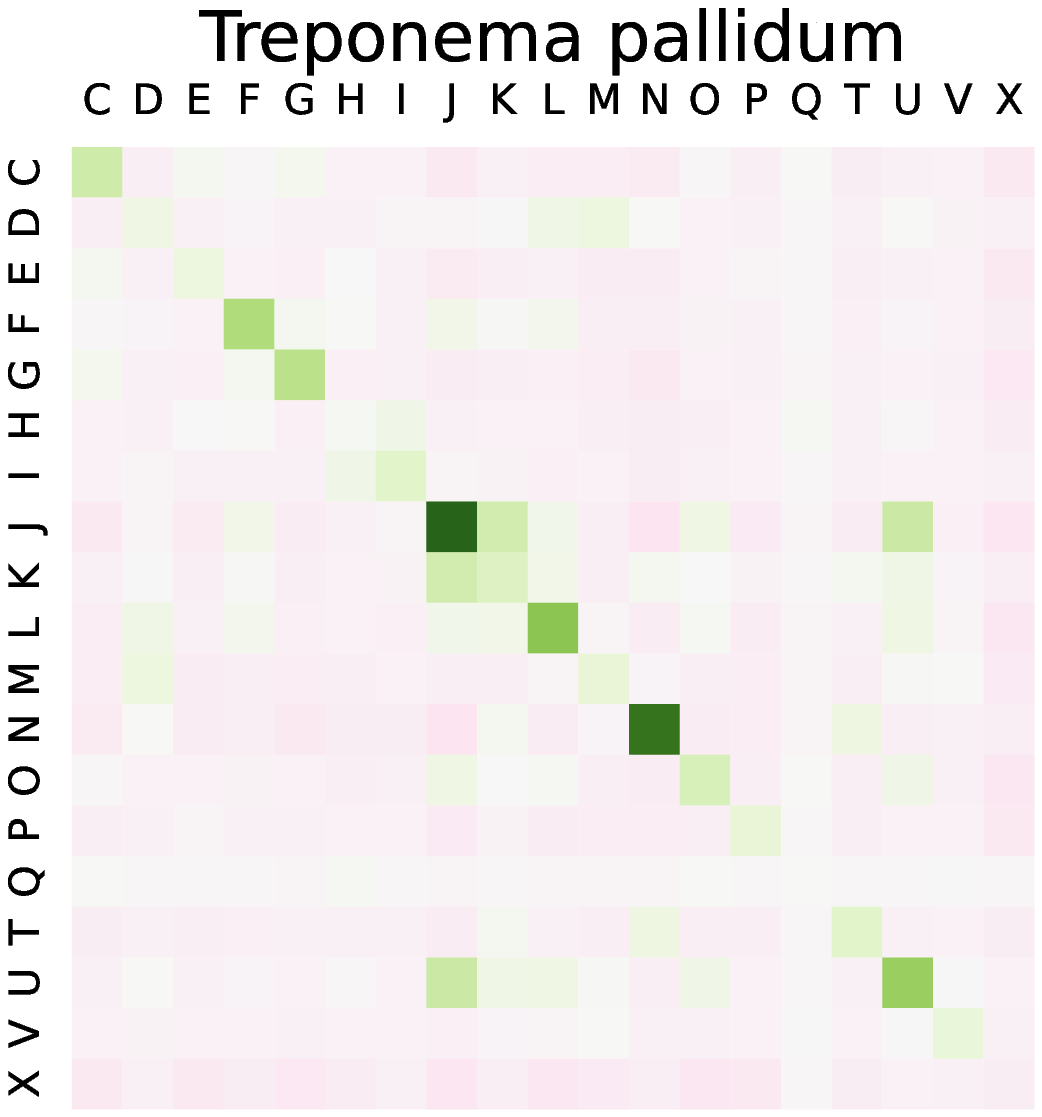}
\includegraphics[scale=0.4]{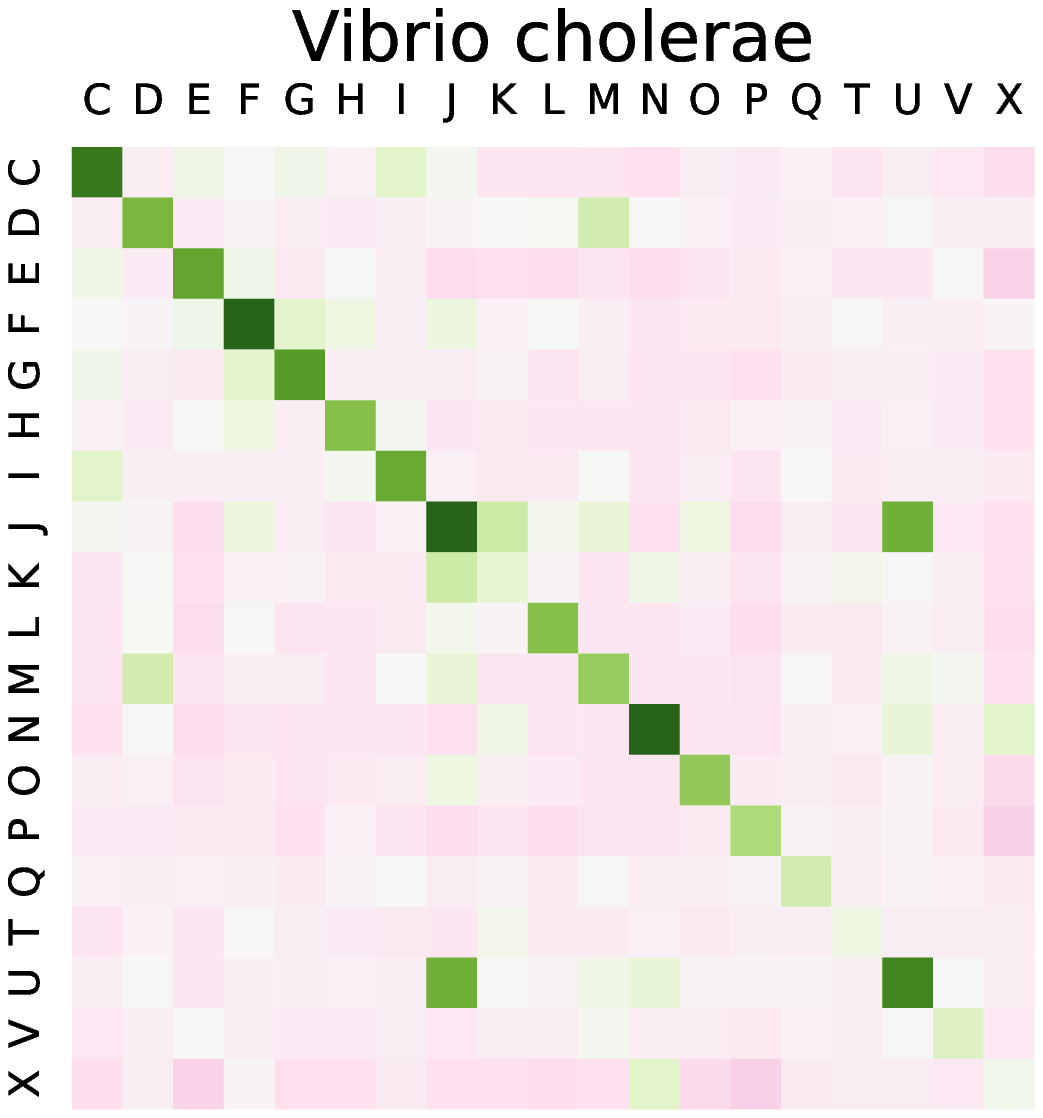}
\includegraphics[scale=0.4]{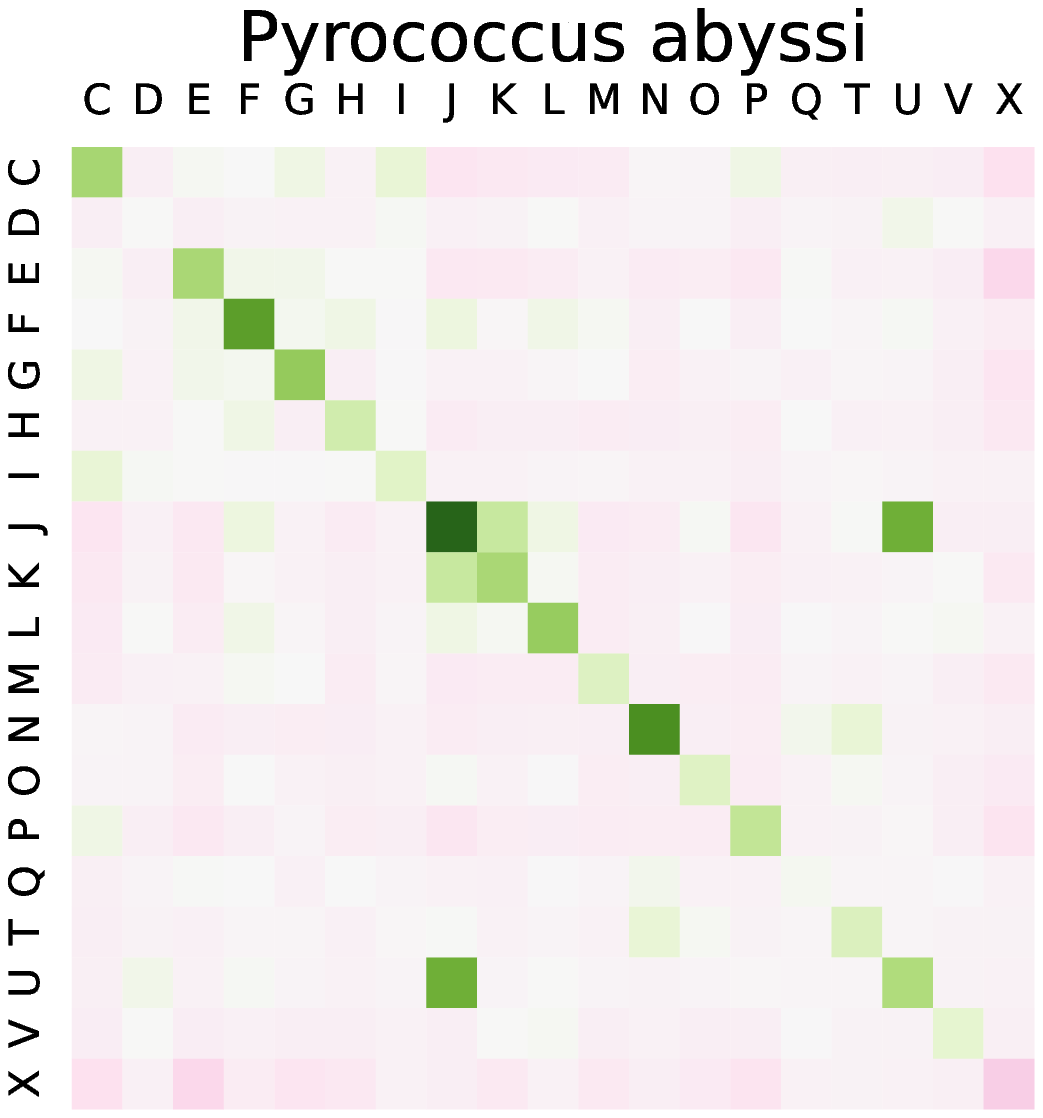}
\includegraphics[scale=0.3]{colorbar.eps}
}

\centerline{
\includegraphics[scale=0.4]{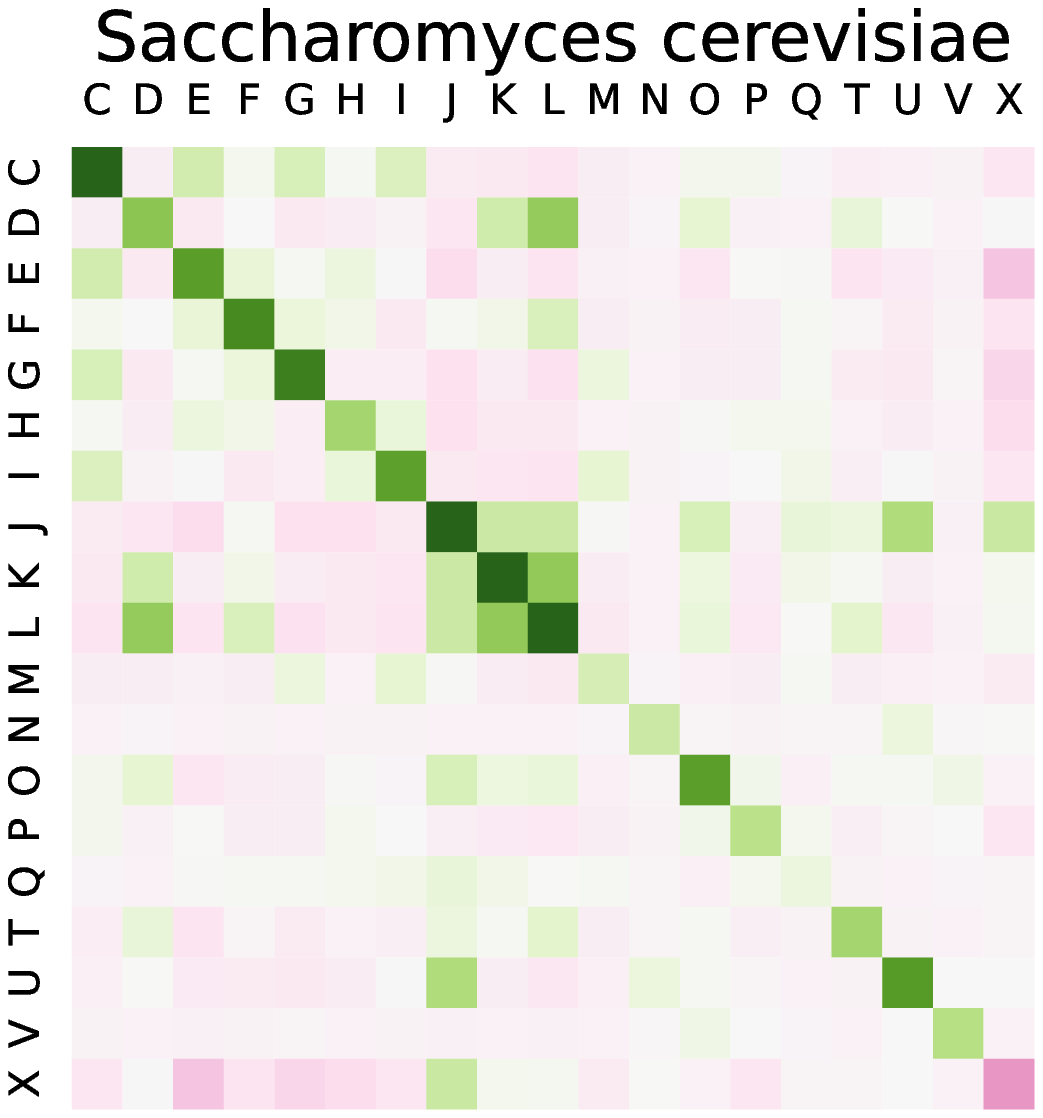}
\includegraphics[scale=0.3]{colorbar.eps}
}

%\caption{Heat-map plots showing $\omega_i$ (along the diagonal) and $\eta_{i,j}$ (for $i \not=j$) ratios for $i$ and $j$ belonging to all considered functional classes, are reported for the ten organism included in the study.}
\caption{Heat-maps corresponding to $\mathbf{Z}$ matrices of the ten organism PPI networks. Diagonal entries correspond to intra-community edges \zscore s, while off-diagonal  entries correspond to inter-community edges \zscore s. Values in the color scale have been cut to  interval $[-10, 60]$.}

\label{fig:heat-mapsPPI}
\end{figure}

%We apply the method described in \eqref{algo}. In details: the density and the number of isolated nodes were computed for the subgraphs induced by each given protein function class $i$ (as reported in Tab.~\ref{tabfunction}) in the PPI network of each organism $o$ (as reported in Tab.~\ref{taborganisms}.
%Those statistics were then compared with the expected ones through means and standard deviations, computed as stated in Theorem 1, providing the \zscore\ values:
  
%$$ Z(\omega_i) = \frac{ m_{i,i} - \mbar_{i,i}} { \sigma^2_{i,i} } $$ 
%$$ Z(L^i) = \frac{ L^i - \Ea{L^i}} { var(L^i) } $$

%For the sake of simplicity we will omit the index $o$ related to the given organism.
 
%In Fig.~\ref{figdensity} and Fig.~\ref{figsingleton} the \zscore s related to the density and the number of isolated nodes are reported respectively for each functional class (reported on the $x$-axis) and each organism (reported with different colors).

Concerning the off-diagonal entries of $\mathbf{Z}$ (namely, those corresponding to inter-community edges) it is worth noting that some classes show significant values, highlighting a unexpected heterophily although in most cases the associated classes belong to close functional classes such as class J, K and L, that can be grouped in the higher category {\it Information, storage and processing}.

In particular significant heterophilic \zscore s are reported, in most of the organism networks, for classes J-L and class J-U representing {\it Translation, ribosomal structure and biogenesis} (class J), {\it Replication, recombination and repair} (class L) and {\it Intracellular trafficking, secretion, and vesicular transport} (class U).
These heterophilic relationships can be considered reasonable from a biological point of view, since nodes associated to protein synthesis in the ribosome (class J) are related to nodes involved in DNA replication (class L) and also to intra-cellular transport (class U) according to the mechanics of protein biosynthesis (when DNA is transcribed, the resulting RNA copy is transported to the ribosome and after translation the protein can be transported away from the ribosome and onto the relevant part of the cell).
These results provide consistency to our work as a real-world validation of our method.

%Interestingly, another significant \zscore\ highlights heterophily in most of the organism networks with respect to classes J and U representing {\it Translation, ribosomal structure and biogenesis} (class J) and {\it Intracellular trafficking, secretion, and vesicular transport} (class U) respectively. 

To have a global and comparative glimpse of the whole scenario concerning PPI, we isolated the diagonal entries of $\mathbf{Z}$ and plotted them in Figure~\ref{fig:density} on a different scale.

%The \zscore s of the density are also reported in log-scale (\ref{figdensity_log}) in order to have a more clear view of the distribution of significant values (it is worth noting that negative values were not shown in log-scale plot).
A large majority of \zscore s (diagonal) shows very high values corresponding to extremely significant deviation from expected ones.
As expected, the exception regards last column related to X class ({\it Function unknown or General function prediction only}) showing \zscore\ values typically negative including very small values (-14 for {\it Saccharomyces cerevisiae}, -7 for {\it Pyrococcus abyssi} and -6 for {\it Escherichia coli}) with only two organisms showing positive values (0.44 for {\it Mycobacterium tubercolosis} and 1.9 for {\it Vibrio cholerae}) (Fig.~\ref{fig:density}).
This typical scenario is consistent with what we could expect from a biological point of view, since it is reasonable that proteins, envolved in a common task, could on average preferentially interact or be close to each other in the PPI.
Proteins belonging to X class do not share a common task since in most of cases they are not associated to any given functional class, so it is reasonable that they are not likely to interact with each other.
Some functional classes seem to show extremely high values, shared among almost all the organisms.
It is evident for class J ({\it Translation, ribosomal structure and biogenesis}) showing the highest values, reaching huge \zscore s (335 for {\it Escherichia coli}, 280 for {\it Mycobacterium tubercolosis}) always higher than 124.
Also class N (Cell motility) shows extremely high \zscore\ values reaching 229 for {\it Brucella mellitensis} and 206 for {\it Escherichia coli}, with the only exception of {\it Mycobacterium tubercolosis} - 2.47 - that is anyway more than two standard deviations greater than the expected one.  
Genes coding for proteins in bacteria are known to typically occur phisically close on chromosome, according to the operon paradigm, and it was shown, consistently with our findings (see \cite{santoni2009}), that especially genes coding for proteins envolved in translation and cell motility task are very close to each other, favoring their syncronous transcription and the interaction of their protein products.
%In general all classes shows significant values, in particular class F (Nucleotide transport and metabolism), G (Carbohydrate transport and metabolism), H (Coenzyme transport and metabolism), L (Replication, recombination and repair), O (Posttranslational modification, protein turnover, chaperones) and U (Intracellular trafficking, secretion, and vesicular transport) show \zscore\ values greater than 10 for all the organisms.
%Some exceptions concern class D (0.20 for Pyrococcus abyssi), K (1.06 for Streptococcus pneumoniae), Q (-0.43 for Streptococcus pneumoniae, -0.14 for Treponema pallidum and 0.87 for Pyrococcus abyssi), T (3.33 for Vibrio cholerae) and V (4.97 for Treponema pallidum).
 
\begin{figure}[H]
\centerline{\includegraphics[scale=1.4]{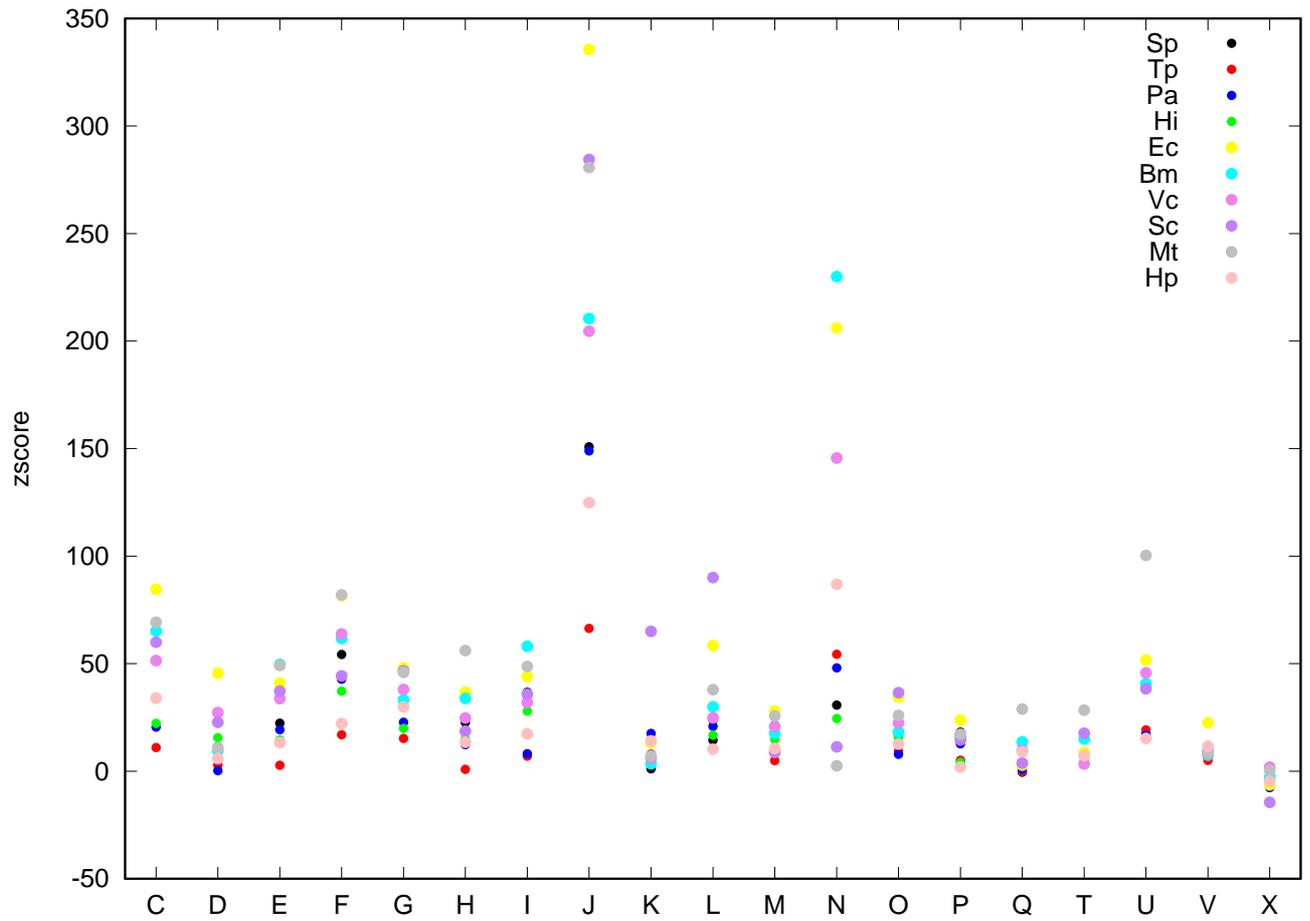}}
\caption{\zscore\ intra-community density values (diagonal entries of $\mathbf{Z}$) of each functional class ($x$-axis) are reported in different colors (each color representing a different organism as indicated in the top right legend of the plot).}   
\label{fig:density}
\end{figure}

%\begin{figure}[H]
%\centerline{\includegraphics[scale=1.4]{zscore_edges_log.eps}}
%\caption{Density log-scale}
%\label{figdensity_log}
%\end{figure}

As for the Pokec social network, results are presented in a completely analogous manner: see Fig.~\ref{fig:heat-map_pokec} for the heat-maps, while in (\ref{fig:density_pokec}) we isolated the diagonal elements. 

\begin{figure}[H]
\centerline{
\includegraphics[scale=0.6]{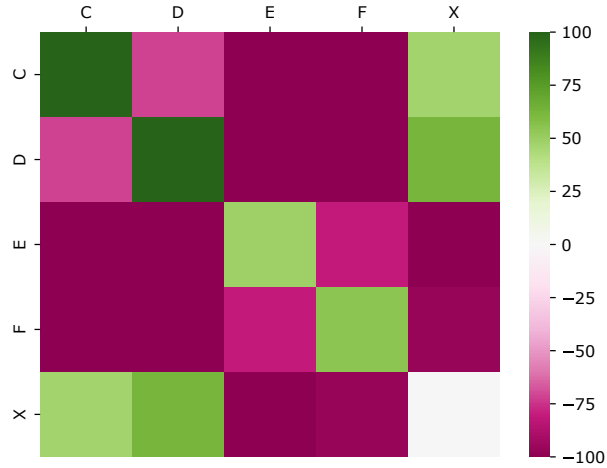}
}
\caption{Heat-map corresponding to $\mathbf{Z}$ matrix of Pokec social network. Values in the color scale have been cut to  interval $[-100, 100]$.}
\label{fig:heat-map_pokec}

\end{figure}

As expected Pokec shows a significant homophilic beahavior with respect to the considered node attribute, age class, as reported in Tab.~\ref{tab:pokec}.

All diagonal \zscore s, excepting class X (no age or non reliable value), reported in Fig.~\ref{fig:heat-map_pokec} and in Fig.~\ref{fig:density_pokec} show highly significant positive values, ranging from an astonishing value around 500 for class C ([12-18) years old) and around 200 for class D ([18-25) years old) till around 50 for classes E ([25-40) years old) and F ([40-60) years old).
Diagonal \zscore\ associated to class X is very close to 0, meaning that users that do not report their age (or report a non reliable age) do not interact with each other.
They prefer to have relationships with other users reporting an age belonging to class C and D (showing positive values in the heat-map Fig.~\ref{fig:heat-map_pokec}), while they do not interact with users belonging to class D and E. 
It can be hypothesized, if we trust in the homophilic nature of social network with respect to age, that most of those users (not reporting their age) have an age belonging to classes C and D.

\begin{figure}[H]
\centerline{\includegraphics[scale=0.9]{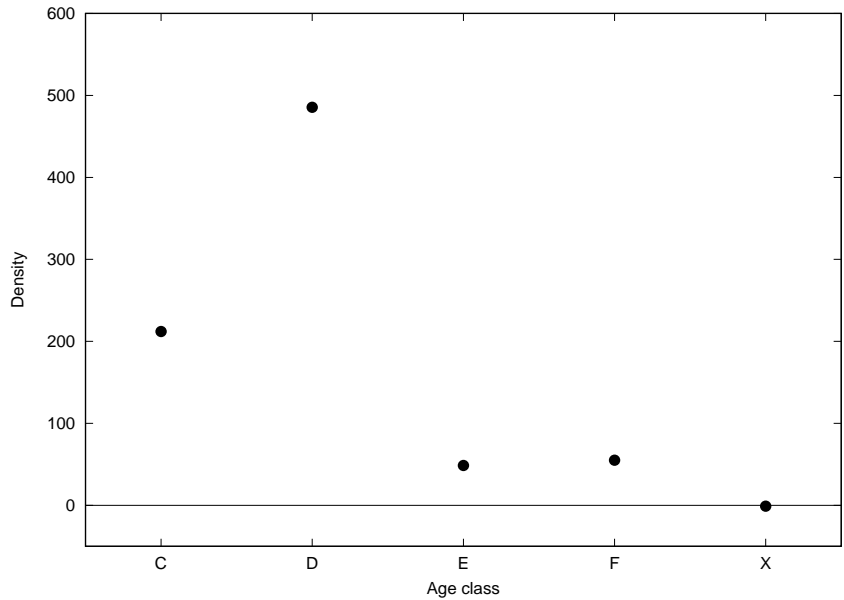}}
\caption{Diagonal \zscore\ values related to age class.}
\label{fig:density_pokec}
\end{figure}

To complement the analysis, we also computed vector $\mathbf{z}_0$. 
Recall that the $i$-th entry of such vector is the \zscore\ of the number of isolated nodes in the subgraph induced by color $i$ (functional class for the PPI and age class for Pokec). 
As explained in Section~\ref{sec:main}, although correlated with the intra-community densities (as confirmed for PPIs in Fig.~\ref{figcorrelation}: the higher the density, the lower the likelihood to find isolated nodes), the entries of $\mathbf{z}_0$ provides a measure of the concentration of the intra-community edges within color classes and, as expected, they are typically negative, consistently with what they represent.
A negative entry means that subgraph induced by the corresponding functional classes for the PPI and age class for Pokec contains less isolated nodes that expected. 
As can be observed in both Fig.~\ref{figsingleton} and Fig.~\ref{figsingleton_pokec}, except for the X class which shows a \zscore\ value close to zero for Pokec and few values close to zero and a vast majority of positive values for PPI, \zscore s associated to all other classes assume very low (negative) values (around 75\% of values are smaller than -5 in PPI - values around -140 for class C and D and around 80 and 60 for class E and F in Pokec), that can be considered extremely significant from a statistical point of view.

\begin{figure}[H]
\centerline{\includegraphics[scale=1.4]{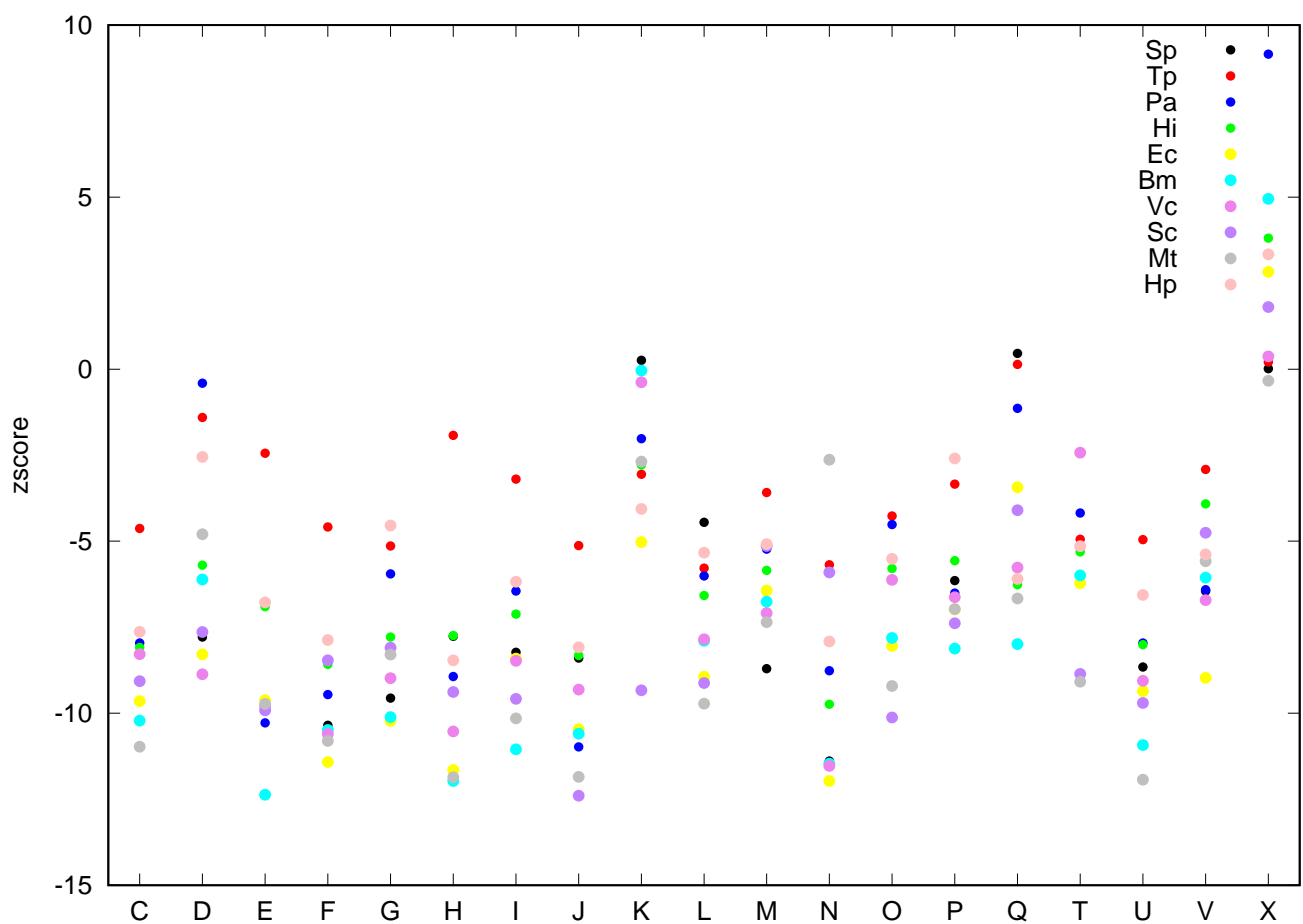}}

\caption{$\mathbf{z}_0$ values (y-axis) of each functional class (x-axis) are reported in different colors (each color representing a different organism as indicated in the top right legend of the plot).}
\label{figsingleton}
\end{figure}

\begin{figure}[H]
\centerline{\includegraphics[scale=0.9]{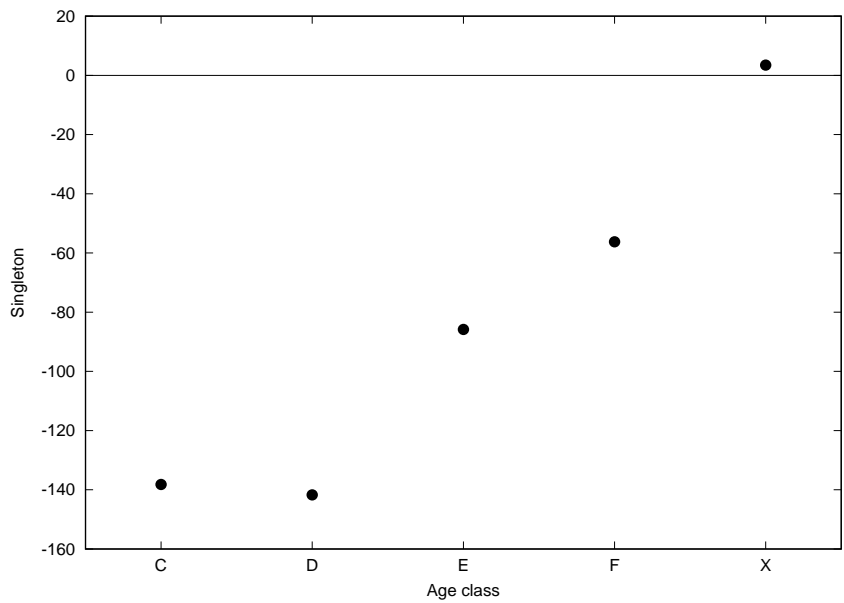}}

\caption{$\mathbf{z}_0$ values (y-axis) of each age class (x-axis) are reported.}
\label{figsingleton_pokec}
\end{figure}

\begin{figure}[H]
\centerline{\includegraphics[scale=0.9]{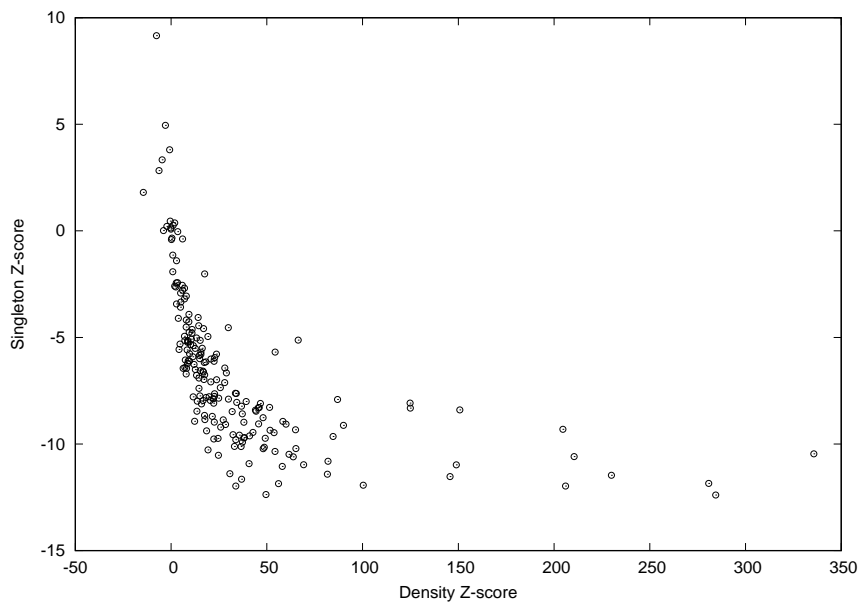}}

\caption{Correlation between diagonal $\mathbf{Z}$ values (x-axis) and $\mathbf{z}_0$ values (y-axis). Each point represents a given functional class of a given organism.}
\label{figcorrelation}
\end{figure}

Finally, as we said in Sections \ref{sec:main} and \ref{sec:measure}, the entries of the $p$-values arrays $1/\mathbf{Z}^2$ and $1/\mathbf{z}_0^2$ (obtained simply by squaring the reciprocal of the entries of the \zscore\ arrays)  can be rather loose estimates of the corresponding true quantiles. In this respect our method is rather conservative. Nonetheless, as shown in Fig.~\ref{figPvalue}, a large majority of $p$-values entries are under the threshold of 0.05, which is usually considered as reliable (for individual testing) with the exceptions already discussed above.

\begin{figure}[H]
\centerline{\includegraphics[scale=1.4]{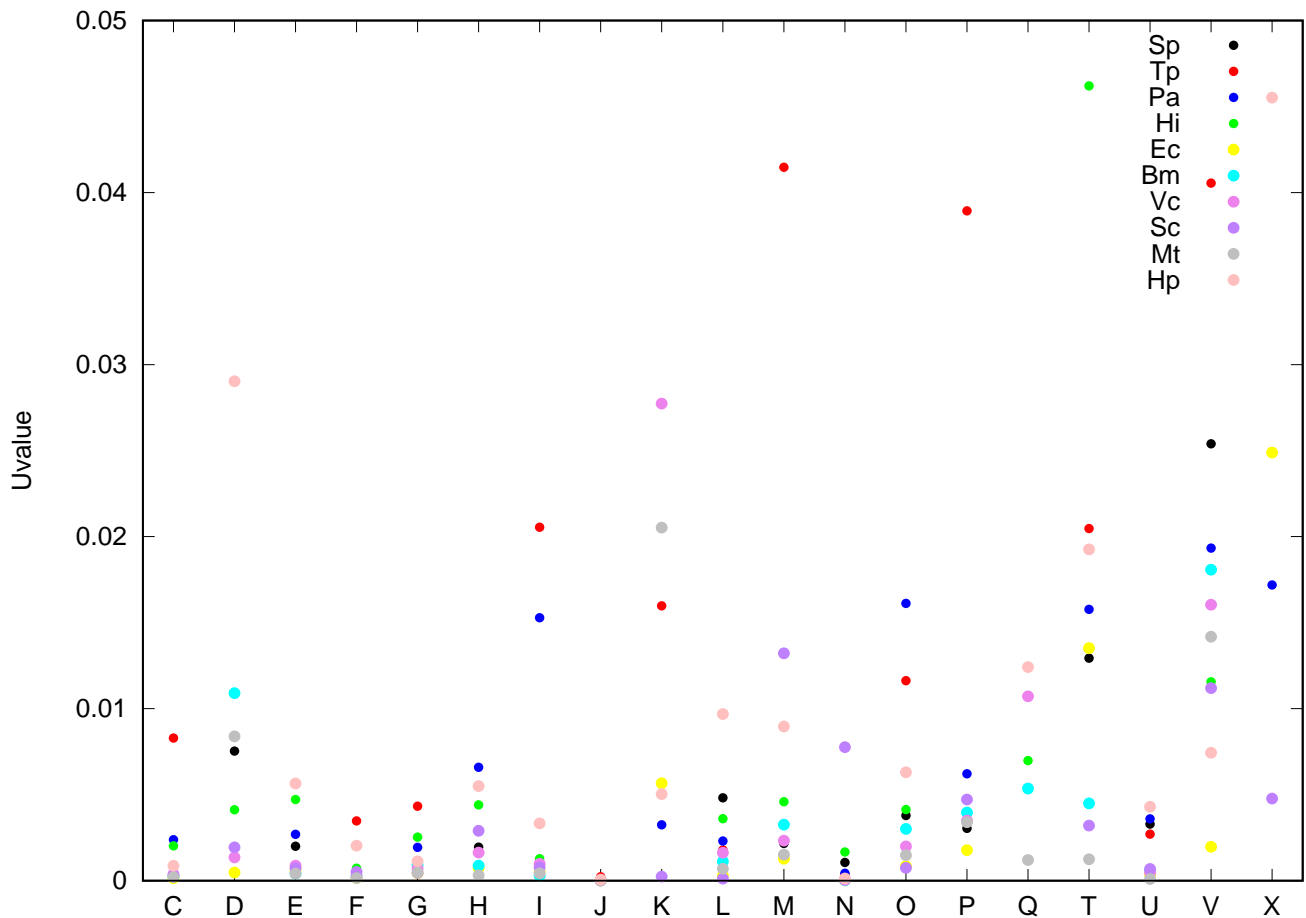}}

\caption{Diagonal entries of $U$-value arrays. The $x$-axis is labelled by functional classes and each color represents a different organism as indicated in the top right legend of the plot.}
\label{figPvalue}
\end{figure}

\section{Conclusions and discussion}\label{sec:danielsky}
In this paper we presented a new approach to assess and measure homophily in networks. The model, described in Section~\ref{sec:measure}, relies on computing
\begin{itemize}
\item the \zscore s of $m^{i,j}$ , the number of edges with one endpoint in functional class $i$ and the other endpoint in functional class $j$ (with possibly $i = j$),
\item the \zscore s of $l^{i}$, the number of nodes in functional class $i$ with no neighbours in class $i$,
\end{itemize}
under the hypothesis that these numbers are samples from the corresponding random variables $M^{i,j}$ and $L^{i}$ under the random coloring model  $(\Phi(\col),\pd_{n,\col})$ (the null model). These \zscore s are either directly interpreted as a refined measure of network homophily (through heat-maps) or serve as the basis either for more synthetic measure via multiple testing or via the significance level of the Euclidean distance between the observed intra-community densities and the expected ones under the random coloring model. The idea of random coloring is implicit in~\cite{park2007} from which we also borrowed terminology. As a result, we extended their model to an arbitrary number of colors and made it computationally efficient and also quantitative (via the \zscore).
The method is clearly applicable to any kind of network and to any of its functional description. Different networks with the same functional description can also be compared directly. Moreover, we noticed that the coefficients of variations of the $M^{i,j}$’s and $L^{i}$'s are invariant for the pair $(G, c)$, where $G$ is the network and $c$ is the profile of the functional description $g$ of $G$.

Obtained results provide evidence of the strong homophilic nature of PPIs, in terms of protein function, and of Pokec social network, in terms of age classes, making our method reliable and affordable since homophilic nature of PPIs and social networks is something expected and known to some extent.

Network homophily is directly linked to network communities and to the paradigm of Guilt By Association (GAS)~\cite{oliver2000}. According to this paradigm, attribute of a given node can be inferred by analyzing the attributes of its neighbours~\cite{deng2003,piovesan2015}.
%It was anyway highlighted how GAS requires a deep analysis, since information coming from neighbours has to be evaluted selecting those critical interactions that could have a major role in determining node attribute (referring to PPIs~\cite{gillis2012}).
In this view assessing and measuring network homophily can be extremely significant for the applicability of the GAS paradigm, allowing to classify nodes according to neighbor attributes.
The analysis of $\mathbf{Z}$ matrix in Pokec network can provide an example of how GAS paradigm can be concretely applied. Users belonging to $X$ class (age not reported or non reliable) are significantly close (according to the values of entries of $\mathbf{Z}$ matrix) to classes $C$ and $D$, showing an heterophilic behavior while they are not close to users of classes $D$ and $E$. This leads to hypothesize that users of class $X$ could have, even if they did not report it, an age associated to class $C$ or $D$.
It is worth noting  anyway that in some networks, in particular in PPIs, node attributes can be already classified through GAS paradigm, leading to a bias or to a tautological analysis, generating a circular argument.

Concerning PPI networks, comparison of $\mathbf{Z}$ matrices shows that the homophilic behavior is not linked to evident stronger similarity among close related species (also {\it Saccharomyces cerevisiae}  and {\it Pyrococcus abyssi} show similar homophilic/heterophilic \zscore s), so that homophilic behavior can be considered as an intrinsic characteristic of PPIs.
Interestingly, some functional classes are more associated than  expected showing an heterophilic behavior, especially classes  $J$, $K$ and $L$, that can be grouped in the higher category ``Information, storage and processing''. 
Another significant \zscore\ highlights heterophily in most of organism networks with respect to classes $J$ and $U$ representing ``Translation, ribosomal structure and biogenesis'' (class $J$) and ``Intracellular trafficking, secretion, and vesicular transport'' (class $U$) respectively. 

The model has been implemented in Python, and experimental results confirm that the computational complexity of the proposed model is optimal for edge density computation, requiring $O(n+m)$ time to compute the $\mathbf{Z}$ matrix. Computing the \zscore\ of the number of $i$-isolated nodes is more time consuming, requiring $O(nm)$ time, but experiments show that it is still efficient in practice for sparse large networks.

In conclusion we are confident that this work can provide a significant contribution allowing to assess and measure, through a robust statistical method, homophily in networks.

\section*{Acknowledgements}
The authors are grateful to an anonymous referee for suggesting a deep revision of the article aimed at putting our methodology at the center, for suggesting one of the methods in the new Section \ref{sec:measure}, and also for inspiring the entire section. We thank both referees for their thorough reading and suggestions that substantially improved the presentation. Finally, we thank Margherita Notarantonio for her help in  result visualization.

\section*{Appendix}\label{appendix}
\subsection*{Basic statistic tools}
In this section, after some preliminaries, we show how to compute probabilities of  events related to random $\col$-colorings and used throughout the paper.
\mybreak
For a positive integer number $a$ and a nonnegative integer number $r$, the \emph{$r$-th falling factorial of $a$} (also referred to as \emph{$r$-th falling power of $a$} in Knuth's terminology [14]) is the number:
$$\kn{a}{r}=\frac{a!}{(a-r)!}=\underbrace{a(a-1)\cdots (a-r+1)}_{\text{$r$ factors}}$$
and it counts the number of injective mapping from a set of $r$ elements into a set of $a$ elements. One has 
\begin{itemize}
	\item[--] $\kn{a}{r}=0$ if $r>a$;
	\item[--] $\kn{a}{0}=1$, $\kn{a}{1}=a$ and $\kn{a}{a}=a!$;
	\item[--] $\kn{a}{r+s}=\kn{a}{r}\kn{(a-r)}{s}$
\end{itemize}
The reason for Knuth's ``falling power" terminology in now clear. Let us come back to the definition of $\col$-coloring which we recall here: let $V$ be a set with $n$ elements and, for a positive integer $s$, let $\col=(c_1,\ldots,c_s)$ be a weak composition of $n$, namely an order sensitive non negative integer vector whose entries add up to $n$. A $\col$-coloring of $V$ is a surjective map $f: V\rightarrow [s]$ such that, for each $i\in[s]$ each \emph{color class} $f^{-1}(i)$ has exactly $c_i$ elements; $\col$ is the \emph{profile of $f$}. The \emph{multinomial coefficient} with \emph{parts} $c_1,c_2\cdots c_s$ 
$$\left( n \atop \mathbf{c}\right)=\left( n \atop c_1c_2\cdots c_s\right)=\frac{n!}{c_1!c_2!\cdots c_s!}$$
counts the $\col$-colorings of $V$. Indeed, the $c_1$ elements that are mapped to $1$ can be chosen in $\left(n\atop c_1\right)$, the elements that are mapped to $2$ can be chosen in $\left(n-c_1\atop c_2\right)$ among the remaining $n-c_1$. Continuing in this way and taking the product of these binomial coefficients we obtain the expression above. Note that, for $s=2$, the multinomial coefficient with parts $c_1$ and $c_2$ (with $c_2=n-c_1$), reduces to the binomial coefficient: 
$$\left(n\atop c_1c_2\right)=\left(n\atop c_1\right)=\left(n\atop c_2\right).$$
Also recall that the binomial coefficient $\left(n \atop r\right)$ is defined for any pair of positive integers $n$ and $r$ as follows,
\[\left(n \atop r\right)=
\begin{cases}
	\frac{n!}{r!(n-r)!} & \text{if $0\leq r\leq n$}\\
	0 & \text{otherwise}.
\end{cases}
\]
Let $J\subseteq [s]$. The \emph{contraction by $J$} of vector $\mathbf{c}=(c_1,\ldots,c_t)$ is the vector $\col'$ obtained from $\col$ by suppressing the entries whose indices are in $J$. We make use of the following multinomial identity which follows straightforwardly by the definition of the multinomial coefficient:
\begin{equation}\label{eq:mult_ident}
	\left( n \atop \mathbf{c}\right)=\left(\prod_{j\in J}\left( n \atop c_j\right)\right)\cdot \left( n-\sum_{j\in J}c_j \atop \mathbf{c}'\right),
\end{equation}
where $n$ and $\mathbf{c}$ are as above and $\col'$ is the contraction of $\col$ by $J$. For instance, if  $J=\{1\}$, then $\col'=(c_2,\ldots,c_s)$ and the expression above reads as 
$$
\left( n \atop \mathbf{c}\right)=\left( n \atop c_1\right)\cdot \left( n - c_1 \atop c_2\cdots c_s\right).$$

We now define the notion of random $\col$-colorings in some more depth. 
Let $\Phi(\col;V)$ be the set of all $\col$-colorings of $V$ (in our case $V=V(G)$ for some graph $G$). When $V$ is understood (as we have assumed throughout the paper) the notation is abridged into $\Phi(\col)$. Thus 
$$\Phi(\col)=\left\{f:V\rightarrow [s] \ |\ f\,\text{is surjective}\right\}.$$
We now equip $\Phi(\col)$ with the uniform measure $\pd_{n,\col}$
$$\pd_{n,\col}(f)=|\Phi(\col)|^{-1}=\left( n \atop \mathbf{c}\right)^{-1}$$
and define the \emph{random $\col$-coloring of $V$}, which we denote by $F$, as the dentity map on $\Phi(\col)$, namely the random $\col$-coloring of $V$ is essentially the probability space $(\Phi(\col),\pd_{n,\col})$ itself and it can be visualized as the random object $F$ taking the value $f\in \Phi(\col)$ with probability $\pra{F=f}=\pd_{n,\col}(f)$. A \emph{statistic based on the random $\col$-coloring $F$ of $V$} is simply any measurable function on $(\Phi(\col),\pd_{n,\col})$, for instance, the indicator $X^i_v$ of the event $(F(v)=i)$, for some $i\in [s]$ and $v\in V$, is one of such. Notice that the inverse image of event $(F(v)=i)$ is the set $\{f\in \Phi(\col) \ |\ f(v)=i\}$. This is the essence of our statistical model. 
\mybreak
For our purposes, for some two disjoint subset $A$ and $B$ of $V$ and some color $i\in [s]$, we are interested in the probability of the event that all the elements of $A$ have color $i$ while all those of $B$ have not. Let $\Omega_i(A,B)$ denote this event. Hence
\begin{equation*}
	\begin{split}
		\pra{\Omega_i(A,B)}&=\pra{\left(F(a)=i,\,\forall a\in A\right)\wedge\left( F(b)\not=i,\,\forall b\in B\right)}\\
		& =\left.\Big|\left\{f\in \Phi_\col \ | \ A\subseteq f^{-1}(i) \subseteq V\setminus B\right\}\Big|\middle/\left( n \atop \mathbf{c}\right)\right..
	\end{split}
\end{equation*}
We are also interested in computing the probability of the intersection of two such events for two distinct colors. We summarize these calculations in the next lemma and then we show how to use the lemma for computing the probability of certain simpler events.
\begin{lemma}\label{lem:fund}
	Let $A$, $A'$, $B$, $B'$ be subsets of $V$ and let $a$, $a'$, $b$, $b'$ be their respective cardinalities. Suppose $A\cap B=\emptyset$, $A'\cap B'=\emptyset$, $A\cap A'=\emptyset$ and $B\cap B'=\emptyset$ and let $b''=|B'\cap A|$. Then, for each two distinct colors $i$ and $j$, one has 
	\begin{equation}\label{eq:mainpr}
		\pra{\Omega_i(A,B)\wedge \Omega_j(A',B')}=\left\{\frac{\kn{c_i}{a}\kn{(n-c_i)}{b}}{\kn{n}{a+b}}\right\}\left\{\frac{\kn{c_j}{a'}\kn{(n-c_i-c_j)}{b'-b''}}{\kn{(n-c_i)}{a'+(b'-b'')}}\right\} 
	\end{equation}
\end{lemma}
\begin{proof}
	Since the elements of $A$ have to be mapped to $i$ and those of $B$ have not, the elements that have color $i$ can be chosen in $\left( n-(a+b) \atop c_i-a\right)$ ways. After this choice, we are left with $n-c_i$ elements that have to be assigned to $[s]\setminus \{i\}$ in such a way that all the elements in $A'$ must be mapped to $j$ and those in $B'$ cannot. Among the elements of $B'$ the are possibly some that have been already assigned to $i$. Therefore we can perform the choice in $\left( n-c_i-(a'+b'-b'') \atop c_j-a'\right)$ ways. After this choice has been done, we are left with $n-(c_i+c_j))$ elements that have to be assigned to colors in $[s]\setminus \{i,j\}$, namely with the number of $\col'$-colorings of a set of $n-(c_i+c_j))$ elements where $\col'$ is the contraction of $\col$ by $\{i,j\}$. If follows that 
	$$\pra{\Omega_i(A,B)\wedge \Omega_j(A',B')}=\frac{\left( n-(a+b) \atop c_i-a\right)\left( n-c_i-(a'+b'-b'') \atop c_j-a'\right) \left( n-c_i-c_j\atop \mathbf{c}'\right)}{\left( n \atop c_i\right)\left( n-c_i\atop c_j\right)\left( n-c_i-c_j\atop \mathbf{c}'\right)},$$
	where we used Formula \eqref{eq:mult_ident} at the denominator. One obtains Formula \eqref{eq:mainpr} after simplifying, expanding the binomial coefficients and resorting to the definition of falling factorial.
\end{proof}
The way we use the lemma to compute the probability of certain basic events is to read the events as a special case of the event $\Omega_i(A,B)\wedge \Omega_j(A',B')$ and to plug in the formula the corresponding parameters $a,b,\ldots b''$. Note that, for any $j\in [s]$, by choosing $A'=B'=\emptyset$ (amd $a'=b'=b''=0$ correspondingly) makes the event $\Omega_j(A',B')$ almost sure. Hence
\begin{equation}\label{eq:forisolates}
	\begin{split}
		\pra{\left(F(a)=i,\,\forall a\in A\right)\wedge\left( F(b)\not=i,\,\forall b\in B\right)}&=\pra{\Omega_i(A,B)}\\
		&=\pra{\Omega_i(A,B)\wedge \Omega_j(\emptyset,\emptyset)}\\
		&=\frac{\kn{c_i}{a}\kn{(n-c_i)}{b}}{\kn{n}{a+b}}=\frac{\kn{c_i}{a}}{\kn{n}{a}}\frac{\kn{(n-c_i)}{b}}{\kn{(n-a)}{b}}.
	\end{split}
\end{equation}  

By, taking $B=\emptyset$---and hence $b=0$--has the effect of suppressing the constraint $(F(b)\not=i,\,\forall b\in B)$. Therefore, for instance, 
\begin{equation}\label{eq:prmonset}\pra{F(a)=i ,\,\forall a\in A}=\pra{\Omega_i(A,\emptyset)}=\pra{\Omega_i(A,\emptyset)\wedge \Omega_j(\emptyset,\emptyset)}=\frac{\kn{c_i}{a}}{\kn{n}{a}}
\end{equation}
and, in particular, for any pair of elements $u,\,v\in V$ any color $i\in [s]$,
\begin{equation}\label{eq:pipi}
	\pra{F(u)=i}=\frac{\kn{c_i}{1}}{\kn{n}{1}}=\frac{c_i}{n}\quad\text{and}\quad\pra{(F(u)=i)\wedge (F(v)=i)}=\frac{\kn{c_i}{2}}{\kn{n}{2}}=\frac{c_i(c_i-1)}{n(n-1)}.
\end{equation}
Analogously, since for any pair of elements $u,\,v\in V$ and any two distinct colors $i,\,j\in [s]$, it holds that 
$$(F(u)=i)\wedge (F(v)=j)=\Omega_i(\{u\},\{v\})\wedge \Omega_j(\{v\},\{u\})$$
it follows that two compute the probability of such an event one has to put $a=b=a'=b'=b''=1$ in Formula \eqref{eq:mainpr} to obtain
\begin{equation}\label{eq:pipj}
	\pra{(F(u)=i)\wedge (F(v)=j)}=\frac{\kn{c_i}{1}\,\kn{c_j}{1}}{\kn{n}{2}}=\frac{c_ic_j}{n(n-1)}.
\end{equation}

\subsection*{Proof of Theorem 1}

\begin{proof}
The expected values $\mbar_{i,i}$ and $\mbar_{i,j}$, $i\not=j$ have already been computed. Let us prove the formula for the expected value of $L^i$. By definition $W_v^i$ is the indicator of the event $(F(v)=i)\wedge (F(w)\not=i,\,\forall w\in N_G(v))$, namely the event that $v$ has color $i$ while all of its neighbors have not. Thus, after (4),
%$$\Ea{W_v^i}=\pra{W_v^i=1}=\pra{X^i_v=1,D^i_{N_G(v)}=0}=\frac{c_i}{n}\pra{D^i_{N_G(v)}=0\Big|X^i_v=1}=\frac{c_i}{n}\left(\frac{\kn{(n-c_i)}{\deg_G(v)}}{\kn{(n-1)}{\deg_G(v)}}\right).$$
$$\Ea{W_v^i}=\pra{W_v^i=1}=\pra{X^i_v=1,D^i_{N_G(v)}=0}=\frac{c_i}{n}\pra{D^i_{N_G(v)}=0\Big|X^i_v=1}=\frac{c_i}{n} \cdot \frac{\kn{(n-c_i)}{\deg_G(v)}}{\kn{(n-1)}{\deg_G(v)}}.$$
Hence, by linearity of expectation 
\[
\mathbb{E}\left(L^i\right)=\frac{c_i}{n}\sum_{v\in V(G)}\frac{\kn{(n-c_i)}{\deg_G(v)}}{\kn{(n-1)}{\deg_G(v)}}.
\]
Let us compute the variance of the random variables in 1), 2) and 3). Observe that all such variables are sums of Bernoulli random variables, namely they are of the form $S=\sum_{\nu\in N}B_\nu$ where $N$ is a finite index set and $B_\nu$ is a Bernoulli random variable for each index $\nu\in N$. The variance of $S$ is thus given by
\begin{equation}\label{eq:varianzagen}
	\begin{split}
	\text{var}(S)&=\Ea{S^2}-(\Ea{S})^2=\Ea{\left(\sum_{\nu\in N}B_\nu\right)^2}-(\Ea{S})^2=\\
	&=\Ea{\sum_{\nu\in N}B_\nu}+\Ea{\sum_{\substack{(\nu,\nu')\in N\times N\\ \nu\not=\nu'}}B_\nu B_{\nu'}}-\left(\Ea{S}\right)^2=\\
	&=\Ea{S}\left(1-\Ea{S}\right)+\sum_{\substack{(\nu,\nu')\in N\times N\\ \nu\not=\nu'}}\Ea{B_\nu B_{\nu'}}=\\
	&=\Ea{S}\left(1-\Ea{S}\right)+\sum_{\substack{(\nu,\nu')\in N\times N\\ \nu\not=\nu'}}\pra{{B_\nu=1\wedge B_{\nu'}=1}}
\end{split}
\end{equation}
where we used the fact that $B_\nu=B_\nu^2$ and that $\Ea{B_\nu B_{\nu'}}=\pra{{B_\nu=1\wedge B_{\nu'}=1}}$. Let us first specialize the formula above to $M^{i,i}$ and $M^{i,j}$. Notice that in both cases $N=E(G)$ and that the summation set in the last equality of \eqref{eq:varianzagen} is $E(G)\times E(G)\setminus \{(e,e) \ |\ e\in E(G)\}$. Denote the latter set by $P$. Since two edges $e$ and $e'$ of $G$ can have at most one node in common, it follows that $P=Q\cup R$ where $Q=\left\{(e,e')\in P \ |\ e\sim e'\right\}$ and $R=\{(e,e')\in P \ | \ e\not\sim e'\}$ and where we have written $e\sim e'$ if $e$ and $e'$ share a node and $e\not\sim e'$ otherwise. Clearly $Q\cap R=\emptyset$. Therefore, if $S$ is either $M^{i,i}$ or $M^{i,j}$, the variance of $S$ is 
$$\text{var}(S) = \Ea{S}\left(1-\Ea{S}\right)+\sum_Q\pra{{B_e=1\wedge B_{e'}=1}}+\sum_R\pra{{B_e=1\wedge B_{e'}=1}}.$$
It is clear that $\pra{B_e=1\wedge B_{e'}=1}$ assumes only two values over the set $P$: it assumes the value $a$ on $Q$, and the value $b$ on $R$. Moreover,
since $|P|=(m^2-m)=2\left(m \atop 2\right)$ and since $e\sim e'$ if and only if $e$ and $e'$ spans a $P_3$, it follows that
$$|Q|=2\pi_3(G)=2\sum_{v\in V(G)}\left( \deg_G(v) \atop 2 \right)\quad\text{and}\quad |R|=2\left(m \atop 2\right)-2\pi_3(G).$$
Therefore, the variance of $S$ assumes the following form 
\begin{equation}\label{eq:duevarianze}
\text{var}(S)=\Ea{S}\left(1-\Ea{S}\right)+2\left[\pi_3(G)(a-b)+ \left(m \atop 2\right)b\right].	
\end{equation}
We obtain expressions for the variance of $M^{i,i}$ and $M^{i,j}$ by plugging the expectation of the corresponding variable in the formula above and specializing $a$ and $b$ for $B_e=Y_e^{i,j}$ and $B_e=Y_e^{i,j}$, with $e=uv$ for some nodes $u$ and $v$.  
\mybreak
Let us start with $a$, namely, the value of $\pra{B_e=1\wedge B_{e'}=1}$ when $(e,e')\in Q$. Hence $e\sim e'$. After regarding edges as sets of two nodes, one has $e\sim e'$ if and only if $|e\cup e'|=3$ (recall that the graph is loopless and has no parallel edges). Let $e\cup e'=\{u,v,w\}$ where $u$ is the unique node in $e\cap e'$.
Recall that for disjoint subsets $A$ and $B$ of $V(G)$ we denote by  
$\Omega_i(A,B)$ the event that all the nodes of $A$ have color $i$ while all those of $B$ have not.
Now, if $S=M^{i,i}$, then $B_e=Y^{i,i}_e$ for all $e\in E(G)$, and thus $a=\pra{\Omega_i(\{u,v,w\},\emptyset)}$; else, if $S=M^{i,j}$, then $B_e=Y^{i,j}_e$ for all $e\in E(G)$; in this case observe $a$ is the sum of the probability of two mutually exclusive events: the event that $u$ has color $i$ while the nodes in $\{v,w\}$ have color $j$, namely the event $\Omega_i(\{u\},\{v,w\})\wedge \Omega_j(\{v,w\},\{u\})$, and the the event that $u$ has color $j$ while the nodes in $\{v,w\}$ have color $i$, namely the event $\Omega_j(\{u\},\{v,w\})\wedge \Omega_i(\{v,w\},\{u\})$. Therefore, by Lemma \ref{lem:fund}, one has
\[
a=\begin{cases}
	\frac{\kn{c_i}{3}}{\kn{n}{3}} & \text{if $B_e=Y_e^{i,j}$}\\
	\frac{c_i\kn{c_j}{2}+\kn{c_i}{2}c_j}{\kn{n}{3}} & \text{if $B_e=Y_e^{i,j}$}
\end{cases}.
\] 
Let us compute $b$. In this case $e\cap e'=\emptyset$. Let $e=uv$ and $e'=u'v'$. If $S=M^{i,i}$, then $B_e=Y^{i,i}_e$ for all $e\in E(G)$, and thus $a$ is the probability of the event $\Omega_i(\{u,u',v,v'\},\emptyset)$, namely the probability that all the four nodes have color $i$ under $F$; else, if $S=M^{i,j}$, then $B_e=Y^{i,j}_e$ for all $e\in E(G)$; observe that there are two bipartitions of $\{u,u',v,v'\}$ into sets $A$ and $B$ such that $|A|=|B|=2$ and neither $A$ nor $B$ induces one of the edges $e$ and $e'$. Hence $b$ is two times the probability that all the nodes in $A$ have one of the colors $i$ or $j$ and all the nodes in $B$ have the other color. Hence $b$ is four times the probability of the event that all nodes in $A$ have color $i$ and all nodes in $B$ have color $j$, that is $b=4\pra{\Omega_i(A,B)\wedge\Omega_j(B,A)}$. Therefore, still by Lemma \ref{lem:fund}, one has
\[
b=
\begin{cases}
	\frac{\kn{c_i}{4}}{\kn{n}{4}}& \text{if $B_e=Y_e^{i,i}$}\\
	4\frac{\kn{c_i}{2}\kn{c_j}{2}}{\kn{n}{4}} & \text{if $B_e=Y_e^{i,j}$}
\end{cases}.
\] 
By plugging the values of $a$ and $b$ (as well as the corresponding expected values) in \eqref{eq:duevarianze} one achieves the desidered expressions for $\sigma^2_{i,i}$ and $\sigma^2_{i,j}$. It only remains  to prove the formula for the variance of $L^i$. By specializing \eqref{eq:varianzagen} with $S=L^i$, $N=V(G)$, $B_v=W_v^i$ one gets
\begin{equation*}
	\text{var}(L^i)=\Ea{L^i}\left(1-\Ea{L^i}\right)+\sum_{\substack{(u,v)\in V(G)\\ u\not=v}}\pra{W^i_u=1,W^i_v=1}
\end{equation*}
and since $\pra{W^i_u=1,W^i_v=1}=0$ whenever $u$ and $v$ are adjacent nodes of $G$, it follows that  
\begin{equation*}	
	\text{var}(L^i)=\Ea{L^i}\left(1-\Ea{L^i}\right)+\sum_{\substack{(u,v)\in V(G)\\ u\not=v, uv\not\in E(G)}}\pra{\left(X^i_u=1,X^i_v=1\right)\wedge \left(D^i_{N_G(u)\cup N_G(v)}=0\right)}.
\end{equation*}
Hence, after setting $b(u,v)=|N_G(u)\cup N_G(v)|=\deg_G(u)+\deg_G(v)-|N_G(u)\cap N_G(v)|$, by \eqref{eq:forisolates} with $a=2$ and $b=b(u,v)$ it follows that 
\[
\begin{split}
	\pra{\left(X^i_u=1,X^i_v=1\right)\wedge \left(D^i_{N_G(u)\cup N_G(v)}=0\right)}=\frac{\kn{c_i}{2}}{\kn{n}{2}}\frac{\kn{(n-c_i)}{b(u,v)}}{\kn{(n-2)}{b(u,v)}}
\end{split}
\]
and after plugging this expression in the latter sum we obtain the stated formula. The proof is thus completed.
\end{proof}

\subsection*{Classes' size in organism's networks}
For each organism's network, we report in Table~\ref{tabnodes}  the number of nodes for each functional class, and the total number of nodes. Nodes in classes A, B, Y, and Z are included in the total size, but were not considered in the analisys.

\begin{center}

\begin{table}[H]

\begin{tabular}{|c|r|r|r|r|r|r|r|r|r|r|}
%\hline
%\multicolumn{3}{|c|}{\textbf{Organism}}&\multicolumn{2}{|c|}{\textbf{PPI network}}\\

\hline
\hline
\textbf{Species}&\multicolumn{1}{|c|}{\textbf{Bm}}&\multicolumn{1}{|c|}{\textbf{Ec}}&\multicolumn{1}{|c|}{\textbf{Hi}}&\multicolumn{1}{|c|}{\textbf{Hp}}&\multicolumn{1}{|c|}{\textbf{Mt}}&\multicolumn{1}{|c|}{\textbf{Sp}}&\multicolumn{1}{|c|}{\textbf{Tp}}&\multicolumn{1}{|c|}{\textbf{Vc}}&\multicolumn{1}{|c|}{\textbf{Pa}}&\multicolumn{1}{|c|}{\textbf{Sc}}\\
\hline
\hline
C&158&203&89&64&168&43&34&151&113&133\\
\hline
D&26&27&23&18&38&20&12&32&16&53\\
\hline
E&298&262&136&86&185&132&20&216&120&172\\
\hline
F&60&64&51&33&63&60&21&65&47&80\\
\hline
G&143&215&98&29&104&174&41&135&69&147\\
\hline
H&114&109&65&65&119&43&19&121&58&95\\
\hline
I&90&65&41&38&129&32&17&67&18&81\\
\hline
J&153&131&140&118&141&136&113&159&146&336\\
\hline
K&107&135&69&22&123&104&26&133&74&143\\
\hline
L&110&118&100&81&155&101&58&133&51&131\\
\hline
M&138&156&110&82&99&81&59&144&43&42\\
\hline
N&33&72&6&42&9&5&43&90&27&5\\
\hline
O&110&99&76&62&92&48&42&104&43&212\\
\hline
P&105&137&80&42&103&64&22&133&62&76\\
\hline
Q&34&23&13&8&78&7&1&35&10&23\\
\hline
T&60&56&32&15&70&39&20&77&13&79\\
\hline
U&29&33&23&35&18&17&11&35&10&78\\
\hline
V&35&38&16&24&36&54&7&40&21&9\\
\hline
X&871&2076&440&400&2048&651&328&1281&619&4169\\
\hline
\hline
\textbf{total nodes}&\textbf{2675}&\textbf{4020}&\textbf{1609}&\textbf{1264}&\textbf{3779}&\textbf{1811}&\textbf{894}&\textbf{3153}&\textbf{1564}&\textbf{6157}\\
\hline
\hline

\end{tabular}
\caption{For each organisms, the total number of nodes (classes A, B, Y, Z included) and the number of nodes in each considered functional class.}
\label{tabnodes}
\end{table}
\end{center}

\subsection*{Speeding-up $L^{i}$ computation}

We show how to compute efficiently statistics in point 3) in Theorem~1, in particular the variance expression
\begin{equation}\label{eq:sommatoriabis}
\mathrm{var}(L^i) = \Ea{L^i}\left(1-\Ea{L^i}\right)+\frac{\kn{c_i}{2}}{\kn{n}{2}}\sum_{\substack{(u,v)\in V(G)\\ u\not=v, uv\not\in E(G)}}\frac{\kn{(n-c_i)}{b(u,v)}}{\kn{(n-2)}{b(u,v)}}
\end{equation}

Trivially computing the summation in~(\ref{eq:sommatoriabis}) requires $O(n^{3})$ time.
We show now that the time complexity can be lowered to $O\left(\sum_{u \in V(G)} \deg^{3}_{G}(u)\right)$, that becomes $O\left(\sum_{u \in V(G)} \deg^{2}_{G}(u)\right)$ expected time (with very high probability) if hash-tables are used to represent sets, and falling factorial values $\kn{x}{y}$ are approximated by applying  Stirling formula. Since huge networks are usually very sparse, this represents a deep improvement with respect to the computation based on~(\ref{eq:sommatoriabis}).

We first observe that
$$\sum_{\substack{(u,v)\in V(G)\\ u\not=v, uv\not\in E(G)}}\frac{\kn{(n-c_i)}{b(u,v)}}{\kn{(n-2)}{b(u,v)}}=$$
\begin{equation}\label{termini}
=\ \sum_{(u,v)\in V(G)}\frac{\kn{(n-c_i)}{b(u,v)}}{\kn{(n-2)}{b(u,v)}}\  - \ 2\sum_{uv\in E(G)}\frac{\kn{(n-c_i)}{b(u,v)}}{\kn{(n-2)}{b(u,v)}}\ - \sum_{u\in V(G)}\frac{\kn{(n-c_i)}{b(u,u)}}{\kn{(n-2)}{b(u,u)}}
\end{equation}

The second and third summations in (\ref{termini}) contain respectively only $O(m)$ and $O(n)$ terms. The first summation  in (\ref{termini}) contains $O(n^{2})$ terms, and for each pair $u,v$ a different value of exponent $b(u,v)$ could be needed. This actually only occurs for pairs $u,v$ having some common neighbor, while if all pairs had distance larger than 2 a substantial speed-up could be possible. We actually compute the first summation  in (\ref{termini}) as if all pairs $u,v$ had no common neighbors, so that $b(u,v) = \deg_{G}(u) + \deg_{G}(v)$, and then we fix the correct value for pairs $u,v$ such that $\dist(u,v) = 2$---adjacent pairs have already been taken into account in the second summation.

Let us denote $\deg_{G}(u) + \deg_{G}(v)$ by $b'(u,v)$:

$$\sum_{(u,v)\in V(G)}\frac{\kn{(n-c_i)}{b(u,v)}}{\kn{(n-2)}{b(u,v)}}\  =$$
\begin{equation}\label{eq:secondsummation}
=  \sum_{(u,v)\in V(G)}\frac{\kn{(n-c_i)}{b'(u,v)}}{\kn{(n-2)}{b'(u,v)}}
\ + \sum_{\substack{(u,v)\in V(G)\\ \dist(u,v)=2}}\left(\frac{\kn{(n-c_i)}{b(u,v)}}{\kn{(n-2)}{b(u,v)}}\  -\ \frac{\kn{(n-c_i)}{b'(u,v)}}{\kn{(n-2)}{b'(u,v)}} \right)
\end{equation}

The first summation in~(\ref{eq:secondsummation}) is easily computed by means of the \emph{degree histogram} of $G$, where 
$ \deg^{-1}_{G}(d)$ is the number of nodes having degree $d$ in $G$:

$$\sum_{(u,v)\in V(G)}\frac{\kn{(n-c_i)}{b'(u,v)}}{\kn{(n-2)}{b'(u,v)}} = \sum_{\substack{0 \leq d_{1}\leq n\\ 0\leq d_{2}  \leq n}} \deg^{-1}(d_{1}) \deg^{-1}(d_{2}) \frac{\kn{(n-c_i)}{d_1+d_2}}{\kn{(n-2)}{d_1+d_2}}
$$
and can be computed in $O(m)$ time, since at most $2\sqrt{m}$ distinct degree values may occur in a graph.
The second summation in~(\ref{eq:secondsummation}) can be computed by exploring the neighborhood of each node, since $\dist(u,v) = 2$ if and only if $u,v \in N_{G}(z)$ for some node $z$ and $uv \not \in E(G)$; this can be done in $O\left(\sum_{z \in V(G)} \deg^{2}(z)\right)$, that is much smaller that $n^{2}$ for sparse graphs. It is immediate to see that $O\left(\sum_{z \in V(G)} \deg^{2}(z)\right)$ is the dominating term in computing the value of (\ref{eq:sommatoriabis}).

Experiments have been performed for social networks with over $10^{6}$ nodes and $8\cdot 10^{6}$ edges, for which the number of pairs of nodes $u,v$ at distance 2 was order of $10^{8}$.

\end{document}